\documentclass[onecolumn]{IEEEtran}

\usepackage{amsmath}
\usepackage{amsthm}
\usepackage{amssymb}
\usepackage{amsfonts}
\usepackage{dsfont}
\usepackage{array}
\usepackage{graphicx}
\usepackage{subfigure}
\usepackage{epsfig}
\usepackage{cite}
\usepackage{caption}

\usepackage[boxed]{algorithm2e}
\theoremstyle{plain}
\newtheorem{theorem}{Theorem}

\newtheorem{lemma}{Lemma}

\theoremstyle{definition}

\usepackage{color}
\newcommand{\corr}{\color{black}}

\begin{document}

\title{Towards a Queueing-Based Framework for In-Network Function Computation}

\author{Siddhartha Banerjee\thanks{Siddhartha Banerjee is with the Department of ECE, the University of Texas at Austin. Email: sbanerjee@mail.utexas.edu}, Piyush Gupta\thanks{Piyush Gupta is with the Mathematics of Networks and Communications group, Bell Labs, Alcatel-Lucent. Email: pgupta@research.bell-labs.com} and Sanjay Shakkottai\thanks{Sanjay Shakkottai is with the Department of ECE, The University of Texas at Austin. Email: shakkott@mail.utexas.edu}}

%\author{
%\IEEEauthorblockN{Siddhartha Banerjee}
%\IEEEauthorblockA{Department of ECE\\
%The University of Texas at Austin\\
%Email: sbanerjee@mail.utexas.edu}\\
%\and
%\IEEEauthorblockN{Piyush Gupta}
%\IEEEauthorblockA{Mathematics of Networks and Communications\\
%Bell Labs, Alcatel-Lucent\\
%Email: pgupta@research.bell-labs.com\\
%}
%\and
%\IEEEauthorblockN{Sanjay Shakkottai}
%\IEEEauthorblockA{Department of ECE\\
%The University of Texas at Austin\\
%Email: shakkott@mail.utexas.edu}\\
%}

\maketitle

\begin{abstract}
We seek to develop network algorithms for function computation in sensor networks. Specifically, we want dynamic joint aggregation, routing, and scheduling algorithms that have analytically provable performance benefits due to in-network computation as compared to simple data forwarding. To this end, we define a class of functions, the Fully-Multiplexible functions, which includes several functions such as parity, MAX, and $k^{th}$-order statistics. For such functions we characterize the maximum achievable refresh rate of the network in terms of an underlying graph primitive, the min-mincut. In {\corr acyclic} wireline networks, we show that the maximum refresh rate is achievable by a simple algorithm that is dynamic, distributed, and only dependent on local information. In the case of wireless networks, we provide a MaxWeight-like algorithm with dynamic flow splitting, which is shown to be throughput-optimal. 
\end{abstract}

\section{Introduction}
\label{intro}

In-network function computation is one of the fundamental paradigms that increases the efficiency of sensor networks v\'{i}s-a-v\'{i}s conventional data networks. Sensor nodes, in addition to sensing and communication capabilities, are often equipped with basic computational capabilities. Depending on the task for which they are deployed, a sensor network can be viewed as a distributed platform for collecting and computing a specific function of the sensor data. For example, a sensor network for environment monitoring may only be concerned with keeping track of the average temperature and humidity in a region.
Similarly `alarm' networks, such as those for detecting forest fires, require only the maximum temperature. The baseline approach for performing such tasks is to aggregate all the data at a central node and then perform offline computations; the premise of in-network computation is that distributed computation schemes can provide sizable improvement in the performance of the network. However, from the perspective of designing network algorithms, in-network function computation poses a greater challenge than data networks as the freedom to combine and compress packets, as long as the desired information is preserved, destroys the flow conservation laws central to data networks. The network has a lot more flexibility, so much so as to make quantifying its performance much more challenging \cite{GiridharKumar}.

Our focus in this paper is to develop a queue-based framework for such systems, and use it to design and analyze network algorithms. By network algorithms, we refer to cross layer algorithms that jointly perform the following tasks:
\begin{enumerate}
\item {\it Aggregating} the data at nodes via in-network computation,
\item {\it Routing} packets between nodes, and
\item {\it Scheduling} links between nodes for packet transmission.
\end{enumerate}
Cross-layer algorithms for data networks, although very successful in both theory and increasingly in real-system implementation, are concerned only with the scheduling and routing aspects. Hence, there is a need for a new framework and new algorithms for in-network function computation in sensor networks. Keeping in mind the lessons learnt from the success of data networks, our aim is to design network algorithms that are {\it dynamic} (i.e., the algorithm should not be designed assuming static network parameters, but rather, use the network state to adaptively learn the network parameters), {\it robust} (i.e., the algorithm adapts to temporal changes in traffic and network topology), {\it capable of dealing with a large class of functions} (i.e., if the function being computed by the network changes, then one should only need to make minor changes to the scheduling and routing algorithms), and {\it generalizable to all network topologies}.

Due to the wide range of potential applications, there are many existing models for such networks, and many different perspectives from which they are analyzed. Some representative works in this regard are as follows:
\begin{itemize}
\item[$\bullet$] The pioneering work of Giridhar and Kumar \cite{GiridharKumar} considers the function computation problem from the point of view of the capacity scaling framework of Gupta and Kumar \cite{GuptaKumar}. In particular, they quantify scaling bounds for certain classes of functions (symmetric, divisible, type-sensitive, and type-threshold) under the protocol model of wireless communications and for collocated graphs and random geometric graphs.
\item[$\bullet$] Other papers consider the function computation problem from the point of view of information theory \cite{OrlitskyRoche, MaIshwarGupta} and communication complexity \cite{Kowshik2, Moscibroda}, characterizing various metrics (refresh rate, number of messages, etc.) for different functions in terms of certain properties of the graph, the function to be computed, and the underlying data correlation. All the above works take an idealized `bottom-up' approach to determine the fundamental limits of the problem, and hence are not directly suitable for designing practical network algorithms. 
\item[$\bullet$] Similar in spirit to the above papers, another approach is to study function computation from the perspective of source coding\cite{Appuswamy, Karamchandani}. These works characterize bounds and show the existence of coding schemes for noiseless, wireline networks. As with the previous algorithms, these policies tend to be idealized, using more complex coding-based schemes instead of simple routing and aggregation (we will later show that such simple strategies suffice for optimal in-network computation of a number of functions of interest); further, these papers do not have explicitly defined policies but rather existence results for such policies. 
\item[$\bullet$] In contrast to the `bottom-up' approach of all the above works, Krishnamachari et al. \cite{KrishEstrinWicker} adopt a more `top-down' approach whereby they formulate network models that abstract out some of the complexity while allowing quantification of performance gains (in their case, energy and delay). Their models do not, however, achieve the optimal throughput and also do not allow for the design of dynamic network algorithms.  
\item[$\bullet$] An alternate model of sensor networks is to assume that nodes are capable of in-network {\it compression}, wherein only the compression (and not merging) of flows is permitted. For example, Baek et al. \cite{Baek} consider routing algorithms for power savings in hierarchical sensor networks. Similarly, Sharma et al. \cite{Neelysensor} design energy-efficient queue-based algorithms under the assumption that the only operation allowed in addition to routing and scheduling is compression of packets at the source node.
\end{itemize}
\vspace{-0.3cm}

The queue-based model for data networks has proved to be an essential tool in designing provably-efficient algorithms for such systems. This model has provided a common framework for understanding various aspects of data network performance such as throughput \cite{TassEph, akrsvw}, delay \cite{BuiSrikantStol, Neelydelay}, flow utility maximization \cite{LinShroffSrikant, BuiSrikantStolyar2, Srikantsensor}, network utility maximization \cite{spbp, StolyarGPD}, distributed algorithms \cite{CSMA}, among others (for an overview, refer to \cite{Neelybook}). In addition, these algorithms have been implemented in real systems  \cite{Akyol1, Ryu}, including in sensor networks \cite{Sridharan1}, with good results. However, these algorithms are designed for data networks, and can not exploit any potential benefit from in-network computation. More recently, this framework has been extended to fork-and-join networks with fixed routing\cite{Towsley}, and resource allocation in processing networks\cite{JiaWal}. 

Using fixed routing in a network usually leads to suboptimal operations as the routes may not be designed to optimize the network performance; in general, even choosing the single best fixed route can perform arbitrarily worse than with dynamic routing (see example in Section~\ref{sec:refrate}). Further, static routing is not robust to temporal changes in the network. However, introducing dynamic routing with in-network computation destroys the flow conservation equations that exist in data networks and networks with fixed flows, as the flow out of a node depends both on inflow as well as (dynamic) packet aggregation at that node. Thus, there is a need to come up with a new queue-based framework and algorithms for efficient function computation in sensor networks, and our paper is a step in this direction.

\subsection{Main Contributions}
\label{contributions}
Our main contributions in this paper are as follows:
\begin{itemize}
\item[$\bullet$] We identify a class of functions, the Fully-Multiplexible or FMux functions, for which we provide a tight characterization of the maximum refresh rate with in-network computation, i.e., the maximum rate at which the sensors can generate data such that the computation can be performed in a stable\footnote{By stability, we refer to the standard notion of the existence of a stationary regime for the queueing process \cite{akrsvw, Neelybook}.} manner. More formally, we show that for these functions, if the refresh rate exceeds a certain graph parameter (the stochastic min-min-cut, which we define formally in Section \ref{sec:refrate}), then the system is transient under \textit{any} algorithm, whereas for any rate lower than this parameter, we construct a policy that can stabilize the system.

\item[$\bullet$] Leveraging the results of Massouli{\'e} et al. \cite{MassTwigg} on {\it  broadcasting}, we obtain a wireline routing algorithm for {\it  aggregation} via in-network computation of FMux functions {\corr in directed acyclic graphs}. Our approach is based on the observation that broadcasting and aggregation are in some sense, duals, of each other. More technically, the duality occurs between 'isolation' of packets in aggregation (i.e., a packet does not have neighboring packets to aggregate with) and that of multiple receptions of the same packet (from different neighbors) in broadcasting. By suitably modifying the approach in \cite{MassTwigg}, we are able to develop an in-network aggregation algorithm for which routing is completely decentralized, and simple random packet forwarding and aggregation suffices for throughput-optimality.

\item[$\bullet$] For general wireless networks we develop dynamic algorithms based on a centralized allocation of routes (dynamic flow splitting) and MaxWeight-type scheduling. In particular, we show that loading rounds on trees in a greedy manner (whereby an incoming round is loaded on the least weighted aggregation tree), coupled with an appropriate scheduling rule, is throughput-optimal for computing FMux functions. The analysis of this algorithm is unique in that in addition to an appropriate Lyapunov function, it requires the construction of appropriate tree-packings of the network graph in order to show the throughput-optimality of this routing scheme.

\end{itemize}

Notation: Throughout the paper, we use calligraphic fonts ($\mathcal{Q},\mathcal{A}$, etc.) to denote sets and the corresponding capital letter ($Q,A$, etc.) to denote their cardinality. We interchangeably use $\cup$ or $+$ for adding elements to sets, and $-$ for deleting elements from sets, and sometimes for brevity of exposition, use the element $u$ to denote the singleton set $\{u\}$ when the meaning is clear from the context (in particular, for a set $S$ and element $u$, $S+u\triangleq S\cup\{u\}$). We also use the shorthand notation $[N]\triangleq\{1,2,\ldots,N\}$.

\section{System Model and Function Classes}
\label{sec:model}

In this section we describe the system model we study in the rest of the paper. At a high level, the system consists of a network of $N$ nodes, one of which is the data aggregator and the rest are sensors. Sensor nodes are capable of three tasks: sensing the environment, transmitting to and receiving data from other nodes, and performing computations on the data. The sensors are assumed to \textit{sense the environment in a synchronous manner}, and the overall purpose of the system is to compute a specific function of the synchronously generated sensor data and forward it to the aggregator. Further, the function computation is assumed to be done in a repeated manner, and the metric used to quantify the efficacy of an algorithm is the maximum synchronous rate at which the sensors can generate data such that the required function of the data can be forwarded to the aggregator in a stable manner. This rate is henceforth referred to as the \textit{maximum refresh rate} of the network.

Before we describe the queueing framework for function computation, we first outline the general communication model that we consider in this work. This model is the same as that considered for studying data networks \cite{akrsvw}. In the next section, we will outline the modifications we make in order to capture the in-network computation aspect of a sensor network.

\noindent\textbf{Communication Graph:} 
We model the topology of the sensor network as a directed graph $G(\mathcal{N,L})$, where $\mathcal{N}$ is a set of $N$ nodes, and $\mathcal{L}$ is a set of $L$ directed links which determine the connectivity between nodes. There is a special node, $a \in \mathcal{N}$, referred to as the {\it  aggregator}, and the rest of the nodes in $\mathcal{N}$ are sensor nodes. Directed link $(u,v)\in \mathcal{L}$ represents that there exists a communication channel from node $u$ to node $v$ (in wireline this corresponds to a physical channel, while in wireless it represents the fact that the nodes are in radio range).
%
%The nodes are partitioned into $2$ sets: $\mathcal{N}$ is the set of sensor nodes and $\mathcal{N}_a$ is the set of aggregator %nodes. For convenience, we henceforth assume that there is a single aggregator node $a$ (however the results can be generalized %to multiple aggregator nodes in a straightforward manner). 
%

\noindent\textbf{Transmission Model:} 
Following the convention in literature \cite{akrsvw,MassTwigg}, we consider a continuous time system in case of wireline systems, whereas in the case of wireless networks, we assume that time is slotted
%\footnote{These assumptions ensure that the results presented here align with existing work, and are not crucial to the analysis. We can convert either system to discrete or continuous time}
, and state all rates in bits per slot. In wireline networks, we define a vector of link rates $\mathbf{\hat{c}}=\left\{c_{uv}\right\}_{(u,v)\in \mathcal{L}}$; one bit is assumed to traverse a link $(u,v)\in\mathcal{L}$ with a random transit time with distribution \textit{Exponential}$(c_{uv})$. The transit times are independent across links and across packets crossing the same link.

For wireless networks, we make the following assumptions/definitions \cite{spbp}:
\begin{itemize}
	\item[$\bullet$] We assume that the channels between nodes are constant (but can extend to time varying channels with added notation, see \cite{akrsvw}). The wireless nature of the network is reflected in the interference constraints.
	\item[$\bullet$] For transmission schedule $I\in 2^\mathcal{L}$, $\mathbf{c}(I)=\left\{c_{uv}(I)\right\}$ denotes the link-rate vector of transmission rates over the links under the chosen schedule.  
	\item[$\bullet$] $\mathcal{I}\subseteq 2^{\mathcal{L}}$ is the set of valid schedules that obey the interference constraints (henceforth referred to as independent sets). $\mathbf{c}(I)$ is said to be \textit{admissible} if the link-rates can be achieved simultaneously in a time slot. $\Gamma =\{\mathbf{c}(I):I\in\mathcal{I}\}$ is the set of all admissible $\mathbf{c}(I)$ and is assumed to be time invariant as stated above. Further, we assume that $c_{uv}(I)\leq c_{\max}\,\,\,\forall\,\,\,(u,v)\in\mathcal{L}, I\in\mathcal{I}$.
	\item[$\bullet$] $\mathbf{\hat{c}}$ is said to be \textit{obtainable} if $\mathbf{\hat{c}}\in\mathcal{CH}(\Gamma)$, the convex hull of $\Gamma$. An obtainable link-rate vector can be achieved by time sharing over admissible link-rate vectors. 
	\item[$\bullet$] From the definition of the convex hull, we have that for every obtainable rate vector $\mathbf{\hat{c}}\in\mathcal{CH}(\Gamma)$, there exists a probability measure $\pi\in\mathbb{R}^{|\mathcal{I}|}_+$ over $\mathcal{I}$ such that $\mathbf{\hat{c}}\leq\{\sum_{I\in\mathcal{\hat{I}}}\pi(I)c_{uv}(I), (u,v)\in \mathcal{L}\}$. The vectors $\pi$ are called Static Service Split (or SSS) rules\cite{akrsvw}, and represent time sharing fractions between different independent sets in order to achieve the rate $\mathbf{\hat{c}}$.
\end{itemize}

Up to this point, the system is identical to one considered for data networks. To highlight the unique features of a physical sensor network performing function computation (and how they affect the modeling of such a system), we consider the following example. In the process, we also indicate the gains achievable via in-network computation versus data download and processing at the aggregator.

%%%\vspace{-0.3cm}
\begin{center}
  \begin{tabular}{p{114mm}}
    \hline
    \hline
 	\\
  \end{tabular}
\end{center}
\vspace{-0.3cm}

\textit{Example $1$:} Consider a grid of $N$ temperature sensors, with a single aggregator at the center, engaged in recording the maximum temperature over these sensor readings. Each node is connected to its four immediate neighbors in the grid via links with a fixed capacity $c$. Every node senses the temperature \textit{synchronously}, and the aggregator desires the MAX of these synchronous measurements. Suppose the network operates by transferring all the data to the aggregator, and then calculating the MAX offline; the maximum rate at which the measurements can be made is then $\Theta(\frac{1}{N})$, as all the packets must pass through one of the $4$ links entering the aggregator. On the other hand, if we allow in-network computation, wherein nodes on receiving multiple packets can discard all but the one with highest value, then the network can operate at a rate of $\Theta(1)$, as the bottleneck is now the minimum-cut of the graph (again the $4$ links entering the aggregator). In subsequent sections, we show that for certain functions like MAX, and any network, the maximum possible refresh rate can be related thus to minimum-cuts in the network. Further, there are dynamic algorithms that support rates up to the maximum refresh rate.

%%%\vspace{-0.3cm}
\begin{center}

  \begin{tabular}{p{114mm}}
    \hline
    \hline
 	\\
  \end{tabular}
\end{center}
\vspace{-0.3cm}
Keeping this example in mind, we now outline the rest of our system model.

\noindent\textbf{Traffic Model:} We consider a symmetric arrival rate model, where each sensor node senses the environment synchronously at a rate $\lambda$ (the refresh rate of the network). The aim of a network algorithm is to support the maximum possible $\lambda$ while ensuring that the network is stable.

In case of wireline networks, packets are generated synchronously at all nodes $i\in\mathcal{N}$ following a Poisson process with rate $\lambda$. In case of wireless networks, the arrival process $A_i[t]$ in time slot $t$ consists of a random number of packets per time slot generated in a synchronous manner , i.e., $A_i[t]=A_j[t]=A[t]\,\forall\,i,j\in\mathcal{N}$, and further $A[t]$ is i.i.d across time. In this case, we define the refresh rate as
\begin{equation*}
\lambda=\mathbb{E}[A[t]],
\end{equation*}  
and also assume that $A[t]$ has finite second moment\footnote{Note that this assumption is not the most general possible restriction on the input process, but one that we choose for convenience of exposition. For more general conditions on the arrival process, refer to \cite{akrsvw}.} which we denote as $m_A=\mathbb{E}[A[t]^2]$.

We associate all simultaneously generated packets with a unique identifier called the \textit{round} number, which represents the time when the packet was generated. In particular, we follow a scheme whereby we number the packets sequentially in ascending order of their generation times, and updating the round numbers when packets complete being aggregated (Thus the oldest unaggregated packet in the network always has round number $0$ and so forth). This scheme of round number allocation is henceforth referred to as the \textit{generation-time ordering}.

Now in order to develop a queueing model, we need a framework to capture the data aggregation operations. As mentioned before, our primary goal is to explore the benefits of in-network computation versus data-download. To this end, we restrict our attention to a specific class of functions, which we define as the FMux functions, and for which we can exactly quantify the gains from in-network computation. The intuition behind the FMux class is that these functions support \textit{maximum compression upon aggregation}; when two (or more) packets combine at a node, the resultant packet has the same size as the original packets. We now define it formally.

\noindent\textbf{Computation Model:}  We assume that the function $f$ is \textit{divisible}\cite{GiridharKumar}. Formally, we assume that each sensor records a value belonging to a finite set $\mathcal{X}$, and we have a function $f$ of the sensor values that needs to be computed at the aggregator $a$. We use $f_k$ to denote the function $f$ operating on $k$ inputs, i.e., $f_k:\mathcal{X}^k\rightarrow\mathcal{R}(f,k)$, where $\mathcal{R}(f,k)$ denotes the range of function when it takes $k$ inputs. Then the function $f$ is said to be {\it  divisible} if:
\begin{enumerate}
\item $|\mathcal{R}(f,k)|$ is non-decreasing in $k$, and
\item Given any partition $\Pi(S)=\{S_1,S_2,\ldots,S_j\}$ of $S\subseteq [n]$, there exists a function $g^{\Pi}(\cdot)$ such that for any $\underbar{x}\in\mathcal{X}^N$:
\begin{equation*}
f(\underbar{x}_S)=g^{\Pi}(f(\underbar{x}_{S_1}),f(\underbar{x}_{S_2}),\ldots,f(\underbar{x}_{S_j})).
\end{equation*}
\end{enumerate} 
Intuitively, for any partition of the nodes, $f$ can be computed by performing a local computation over each set in the partition, and then aggregating them.

We define a function $f$ to be \textbf{\textit{Fully-Multiplexible or FMux}} if $\mathcal{R}(f,k)=\mathcal{R}(f,j)=\mathcal{R}(f)$ for all $j,k\in [n]$. In other words, the output of a FMux function lies in the same set independent of the number of inputs. Some important examples of FMux functions are MAX, k-th order statistics, parity, etc.. As mentioned before, in this work we will focus on FMux functions as they most clearly exhibit the effects of in-network computation (in that we have tight bounds for their refresh rate). %In Section \ref{sec:genfunc}, we will discuss how to extend this model to general function computation. 

As a representative example of FMux functions for defining the queueing model and algorithms, consider computation of the parity of the sensor readings: $\mathcal{X}=\{0,1\}, f(\{x_1,x_2,\ldots,x_N\})=x_1\oplus x_2\oplus\ldots\oplus x_N$, where $\oplus$ represents the binary XOR operator. Upon sensing, node $i$ stores the value $x_i$ as a packet of size $\log |\mathcal{X}|=1$ bit. Next, when two or more packets of the same round arrive at a node, they are combined using the XOR operation. Finally, the aggregator obtains the parity by taking XOR of all the packets of a given round that it receives. We now develop a queueing model for FMux functions. 

\noindent\textbf{Queueing Model:} Each node maintains a queue of packets corresponding to different rounds. For node $i,\mathcal{Q}_i[t]\in 2^{\mathbb{N}_0}$ is a subset of $\mathbb{N}_0$ representing the round numbers of all packets queued up at that node. We also define $Q_i[t]=|\mathcal{Q}_i[t]|$.

When a packet corresponding to round $r$ arrives at node $i$ from any other neighboring node, it is combined with node $i$'s own packet corresponding to round $r$ to result in a single packet of the same size (using the FMux property in general, e.g. by taking XOR for parity). In the case where node $i$ does not have a packet of round $r$ in queue, it needs to store the new packet. Formally, upon arrival of packet of round $r$ in time slot $t$ (and ignoring other arrivals and departures), the queue is updated as follows-
\begin{equation*}
\mathcal{Q}_i[t+1]=\mathcal{Q}_i[t]\bigcup r\mathds{1}_{\{r\notin \mathcal{Q}_i[t]\}}, 
%%%\vspace{-0.1cm}
\end{equation*}
where we use $r\mathds{1}_{\{r\notin Q_i[t]\}}$ as shorthand for `$r$ if $r\notin \mathcal{Q}_i[t]$, else $\phi$'. The complete queue update in a time slot is obtained by extending this definition for all arrivals, and by removing any departing packets from the queue.

%\begin{figure}[tp]\label{fig:queuemodel}
%  \centering \includegraphics[scale=0.4]{WSNmod.pdf}
%  \caption{The queueing dynamics under FMux computation (packets of different rounds are shaded differently). In (a), an `innovative' packet arrives at a queue, either from a neighbor or due to sensing, and is stored. In (b), a packet arrives from a neighbor and gets combined with another packet of the same round present in the queue. In (c), a packet is transmitted and exits the queue}
%\end{figure}

Suppose further that the round number allocation is done according to the generation-time ordering scheme described before, then the system described above forms a Markov chain under any stationary scheduling and routing algorithm. Further, it can be showed that this chain is irreversible and aperiodic. We will now focus on the above queueing dynamics for the design of scheduling and routing algorithms. 

We should note here that the queueing model we have described above accounts only for routing and aggregation of packets \textit{belonging to the same round}. We have not allowed packets from different rounds to be combined together in any way, thereby negating the possibility of block coding and network coding. In the case of parity, it is known however that there is no improvement possible by using schemes with block/network coding \cite{GiridharKumar,Kowshik1}.

\section{Maximum Refresh Rate and Tree Packing}
\label{sec:refrate}

Given the above queueing model, it is unclear what routing structures are required for efficient in-network computation. Existing routing-based approaches for function computation \cite{KrishEstrinWicker,Towsley,TAG} often assume that routing is done on a single \textit{aggregation tree}, where each node aggregates data from its children before relaying it to its parent. However it's not \textit{a priori} evident that a single optimal tree, or a collection of optimal trees exists (or indeed that acyclic aggregation structures are sufficient), and if it does, how it can be found dynamically.

In this section, we derive an algorithm-independent upper bound on the refresh rate for FMux computation. By focusing on the flow of information from sensor nodes to the aggregator, we are able to express the bound in terms of an underlying graph primitive- the min-mincut of the graph. Next we construct a class of throughput-optimal randomized policies, thereby obtaining a tight characterization of the maximum refresh rate. In the process, we show the existence of an optimal collection of aggregation trees. To understand the import of this result, consider the following example.

%%%\vspace{-0.3cm}
\begin{center}
  \begin{tabular}{p{114mm}}
    \hline
    \hline
 	\\
  \end{tabular}
\end{center}
\vspace{-0.3cm}
\textit{Example $2$:} Consider a wired network $G$ consisting of the complete graph on $N$ nodes, with every edge having capacity $1$. If we use a single aggregation tree for routing, then the maximum possible refresh rate for computing the parity function is $1$, as every edge is a bottleneck. However, by using a collection of aggregation trees (in fact, it can be shown that a particular set of $N-1$ trees are sufficient), one can achieve a refresh rate of $N-1$. This, as we will show in the next section, is optimal as it turns out to match the min-mincut of the graph. 
%%%\vspace{-0.3cm}
\begin{center}
  \begin{tabular}{p{114mm}}
    \hline
    \hline
 	\\
  \end{tabular}
\end{center}
%%%\vspace{-0.3cm}
Keeping this in mind, we now characterize the maximum refresh rate for FMux computation in general graphs.

\subsection{An upper bound on refresh rate for FMux computation}
\label{uppercased}

We now state an upper bound for the refresh rate under which the network can be stabilized under any algorithm. We state this theorem for wireless networks, as an equivalent theorem for wireline networks can be obtained as a special case. 

Given a rate vector $\mathbf{\hat{c}}\in\mathcal{CH}(\Gamma)$ and any node $i\in\mathcal{N}$, we define the min-cut between the node $i$ and the aggregator $a$ as:
\begin{equation*}
\delta_{i}(\mathbf{\hat{c}})=\min_{\{S\subset \mathcal{N}:i\in S,a\notin S\}}\sum_{u\in S,v\notin S}\hat{c}_{uv}.
\end{equation*}
Further, we define the min-mincut of the network under rate vector $\mathbf{\hat{c}}\in\mathcal{CH}(\Gamma)$ as
\begin{equation*}
\delta^*(\mathbf{\hat{c}})=\min_{i\in\mathcal{N}}\delta_{i}(\mathbf{\hat{c}}).
\end{equation*}
Then we have the following lemma.

\begin{lemma}
\label{lem:minmincut}
\noindent\textbf{Upper Bound on refresh rate:} 
Consider a network performing FMux computation. A refresh rate of $\lambda$ can not be stabilized by any routing and scheduling algorithm if
\begin{equation*}
\lambda> (\log |\mathcal{R}(f)|)^{-1}\max_{\mathbf{\hat{c}}\in\mathcal{CH}(\Gamma)}\delta^*(\mathbf{\hat{c}}).
\end{equation*}
\end{lemma}

We note here that the capacities of the links are given in bits per time slot, while the refresh rate $\lambda$ is in terms of packets per time slot. The $(\log |\mathcal{R}(f)|)^{-1}$ factor is to convert link capacities into packets per time slot, and is henceforth present in all bounds for the refresh rate.

\begin{proof}
The proof follows from tracing the steady state flow of packets from any sensor node to the aggregator. More specifically, for a refresh rate $\lambda$, suppose the network is stabilized by some algorithm. Then the Markov Chain described by the packets in the network (under the generation time ordering round number allocation, as described above) has a stationary regime. Further, due to the network constraints, the average service rate on each edge of the network in the stationary regime is given by some $\mathbf{\tilde{c}}\in\mathcal{CH}(\Gamma)$ (in bits per slot)

Next under the stationary regime, for a sensor node $i\in\mathcal{N}$, we can trace the packets as they travel from node $i$ to the aggregator (in order to do so, we start tracing a packet when generated at $i$, and subsequently whenever that packet is aggregated, we trace the aggregated packet). Now for every directed path from $i$ to $a$, we obtain an average flow of packets which travel along that path. This gives us a set of flows from $i$ to $a$. Due to the unchanging packet size (due to the FMux assumption), the sum of these flows is equal to $\lambda$. However, due to the network constraints, the sum of flows on an edge $(u,v)$ is less than or equal to $\tilde{c}_{uv}$, and thus by the max-flow-min-cut theorem, $\lambda$ is less than or equal to the minimum $i-a$ cut with edge capacities given by $\mathbf{\tilde{c}}$. Now since this is true for any node $i$, we have that $\lambda\leq(\log |\mathcal{R}(f)|)^{-1}\delta^*(\mathbf{\tilde{c}})$. Maximizing over all $\mathbf{\tilde{c}}\in\mathcal{CH}(\Gamma)$, we get our result by contradiction.  
\end{proof}

\subsection{An optimal class of randomized scheduling/routing policies}
\label{treepackbnd}

From Lemma \ref{lem:minmincut}, it is evident that the min-mincut of the graph (under an appropriate SSS rule) is the bottleneck for computing an FMux function. Now we can use a classical theorem of Edmonds in order to simplify the space of policies we need to consider. We state the theorem in its original form whereby it is applicable to a one-to-all network broadcast scenario (informally, a directed graph with a special source node, where the aim is to transmit the same packets from the source to all the nodes in the network). However given a sensor network, we can apply Edmonds' Theorem on it by reversing the directions of all the edges while keeping their capacities the same. 

Consider a directed graph $G(\mathcal{N},\mathcal{L})$ with a distinguished source node $s\in \mathcal{N}$, and suppose each edge $(u,v)\in \mathcal{L}$ of the graph is associated with a capacity $c_{uv}>0$. As before, the min-mincut of the graph $G$ is defined as:
\begin{equation*}
\delta^*(G)=\min_{i\in \mathcal{N}\backslash\{s\}}\min_{S\subset \mathcal{N}:i\in S,s\notin S}\sum_{u\notin S,v\in S}\hat{c}_{uv}.
\end{equation*}
Let $\mathcal{T}$ to be a set of all spanning trees of $G$ rooted at $s$ (i.e., every $t\in\mathcal{T}$ is a spanning tree with $s$ as the first element in its topological order). The max-spanning-tree-packing number, $\Lambda^*(G)$ is defined to be the solution to the following optimization problem:\\
Maximize \hspace{1.5cm}$\sum_{\tau\in\mathcal{T}}\lambda_\tau$,\\
subject to
\begin{align*}
\sum_{\tau\in\mathcal{T}:(u,v)\in \tau}\lambda_\tau &\leq c_{uv}\,&\forall\,(i,j)\in \mathcal{L},\\
\lambda_\tau &\geq 0\,&\forall\,\tau\in\mathcal{T}.
\end{align*}
Then we have the following theorem.
\begin{theorem}
\label{thm:edmonds}
\textbf{(Edmonds, $\mathbf{1972}$ \cite{edmonds})} For a directed graph $G(\mathcal{N},\mathcal{L})$ with distinguished source vertex $s$ and edge capacities $c_{uv}, (u,v)\in \mathcal{L}$, the min-mincut $\delta^*(G)$ is equal to the max-spanning-tree-packing $\Lambda^*(G)$.
\end{theorem}

Edmonds' Theorem guarantees the existence of a tree packing which has the same weight as the min-mincut of the graph. Now in the case of one-to-all broadcast in networks, wherein a node can transmit copies of any packet it has received, it is clear that the subgraph traced out by a packet in reaching all nodes forms a tree. Returning to the wireless setting, we now sketch out how to construct a randomized routing and scheduling algorithm that is throughput-optimal, using the technique developed by Andrews et al \cite{akrsvw}. Suppose we know the point $\mathbf{\tilde{c}^*}\in\mathcal{CH}(\Gamma)$ in the obtainable rate region which maximizes the min-mincut\footnote{Note that such an optimal rate point exists as the min-mincut is a continuous function of the rates, which lie in a compact set $\mathcal{CH}(\Gamma)$}, then we can schedule according to the corresponding SSS rule to achieve an ergodic rate of $\tilde{c}^*_{uv}$ across any link $(u,v)$. The network is now converted into a wired network, i.e., with edges having fixed capacities. Next we can use Edmonds' Theorem to obtain a tree packing for this fixed-capacity network, which determines how the input flow should be balanced between spanning trees. Combining these two steps, we obtain a scheme whereby we split the incoming flow according to the tree packing, and schedule using the SSS rule corresponding to $\mathbf{\tilde{c}}^*$ to stabilize the system. By a similar argument, we can obtain a tree packing given the optimal SSS rule for the FMux function computation problem. Here each round is associated to a spanning tree such that the total incoming flow (which is equal to the refresh rate) is split according to the above tree packing. This tree is henceforth referred to as the \textit{aggregation tree} of the round, and determines the route followed by the packets in that round. The routing thus taken care of, the scheduling is done according to the optimal SSS rule, and in combination, they stabilize the network. Combined with Lemma \ref{lem:minmincut}, this gives a tight characterization of the maximum refresh rate of the network, which we state in the following theorem.

\begin{theorem}
\label{thm:maxrefrate}
Consider a network performing in-network computation for an FMux function $f$. The \textit{maximum refresh rate} is defined as:
\begin{equation*}
\lambda^* = (\log |\mathcal{R}(f)|)^{-1}\max_{\mathbf{\hat{c}}\in\mathcal{CH}(\Gamma)}\delta^*(\mathbf{\hat{c}}).
\end{equation*}
Then a refresh rate of $\lambda$ can not be stabilized by any algorithm if $\lambda>\lambda^*$, and there exists a static, randomized algorithm to stabilize it if $\lambda<\lambda^*$.
\end{theorem}

We note that this bound, and the definition of FMux functions, is similar in spirit to the results in \cite{Appuswamy}. Theorem \ref{thm:maxrefrate} is different (and more general) than the results obtained in \cite{Appuswamy}, both in scope and technique. More generally, there is a fundamental difference in the level of abstraction with which we view the problem vis-a-vis other similar works such as \cite{Kowshik2, Moscibroda, Appuswamy}, where the focus is on the physical/link layers, and further, {\it only for wired networks}. Our result is for network layer algorithms for a more general class of networks (wired and wireless); furthermore, the algorithm based on SSS rules is an explicit (albeit static) algorithm, and uses only routing and packet aggregation at nodes. In contrast Appuswamy et al. \cite{Appuswamy} use results that {\it show the existence of source-coding based, schemes} (which are more complex than routing based schemes we use) that achieve the min-mincut in {\it noiseless, wired networks}.

The problem with such a static algorithm is that it needs prior calculation of the min-mincut and associated optimal rate point (to obtain the optimal packing of aggregation trees). A better alternative is to use the queues as proxy for learning these through dynamic algorithms based on the current system state (similar to the Backpressure algorithm\cite{akrsvw, TassEph} for data networks). The rest of the paper deals with the development of such algorithms. 

\section{Routing with Random Packet Forwarding in Wired Networks}
\label{sec:Algo1}

In this section we give a routing algorithm for {\corr acyclic} wired networks based on random packet forwarding with aggregation. This algorithm is based on an algorithm for one-to-all network broadcast in wireline networks by Massouli{\'e} et al.\cite{MassTwigg}, which demonstrates that random `useful'-packet forwarding achieves the min-mincut bound. We modify their approach to obtain a dual version applicable to FMux computation in wireline networks.

In in-network FMux computation, as described before, a new round of packets arrives at all sensor nodes in a synchronous manner, and need to be routed to the aggregator. For the broadcast problem (where packets arrive at the source and need to be routed to all other nodes), an optimal algorithm\cite{MassTwigg} is as follows: for any idle link in the network, the source node randomly picks a packet that the receiver does not have (defined as a 'useful' packet) and transmits it on that link. We now define an analogous notion of a useful packet for in-network aggregation, and show how it can be used to derive an optimal random routing algorithm for FMux function computation.

A natural invariant in broadcast is that the trace of a round of packets {\it always follows a spanning tree}. This is not in general true in aggregation; {\corr however in the case of acyclic networks}, one can impose additional constraints to ensure that a transmission does not lead to an {\it isolated packet}, i.e., a packet at a node such that no neighboring node has a packet from the same round, thus preventing its aggregation. This can be ensured by defining an appropriate notion of a 'useful' packet and {\it only transmitting useful packets}. We define a packet in node $i$ to be useful to neighbor $j$ if {\it  (a)} $j$ has a packet of the same round (hence ensuring aggregation); and {\it  (b)} transferring the packet to $j$ does not result in an isolated neighbor $k$ of $i$. The routing algorithm now consists of randomly forwarding useful packets whenever a link is idle. {\corr In Appendix \ref{app:Algo2}, we prove that this definition leads to packets being routed on spanning trees.}

Formally, the algorithm is a work-conserving policy whereby each node $i\in\mathcal{N}$ ensures that an outgoing edge $(i,j)\in\mathcal{L}$ is engaged in a packet transfer {\it if and only if} there are packets in $i$ that are useful to $j$. For a node $i$, we define $N^-(i)=\{j\in\mathcal{N}:(j,i)\in\mathcal{L}\}$ and $N^+(i)=\{j\in\mathcal{N}:(i,j)\in\mathcal{L}\}$ to be the `in-neighborhood' and `out-neighborhood' of $i$ respectively. 
%Further we define $\mathcal{S}$ to be the collection of all sets $S\subseteq\mathcal{N}$ such that the subgraph induced by $S$ contains a spanning tree rooted at $a$. 
Now at a given time $t$, packets of a round $r$ can be in $3$ states under the algorithm (analogous to the notation Massouli{\'e} et al.\cite{MassTwigg}):
\begin{enumerate}
	\item Sucessfully aggregated, i.e., present only at aggregator $a$.
	\item Idle, i.e., not being transmitted at any edge. Packets of an idle round $r$ are present at nodes of some set $S\subset \mathcal{N}$, henceforth called the {\it footprint-set of round $r$}, and denoted $FP_r[t]$. We define a {\it valid footprint-set} to be one where the subgraph induced by the set contains a spanning tree rooted at $a$ {\corr (equivalently, each node in the footprint set has a directed path to $a$)}; the collection of such sets is denoted as $\mathcal{S}$. Finally, for all $S\in\mathcal{S}, X_S[t]\triangleq $ is a count of idle rounds located in $S$.
	\item Active, i.e., being transmitted on at least $1$ edge. The collection of active rounds is given by $\mathcal{A}[t]\triangleq\{R_1[t],R_2[t],\ldots ,R_m[t]\}$, where round each round $R_k[t]$ has an associated pair $(FP_k[t],E_k[t])\in\mathcal{S}\times 2^{\mathcal{L}}$; here $FP_k[t]$ is the footprint-set, and $E_k[t]\subset \mathcal{L}$ is the set of edges on which packets of round $R_k[t]$ are being transmitted.	
\end{enumerate}
The pair $(\{X_S[t]\}_{S\in\mathcal{S}},\mathcal{A}[t])$ forms a complete description of the system; we henceforth consider the Markov Chain on this system description for describing and analyzing the algorithm. Further, for ease of exposition, we supress the dependence on time whenever clear from context.

Now we can formalize the notion of a \textit{useful packet} for transmission. We define an edge $(u,v)$ to be idle if $(u,v)\notin E_r \forall r\in\mathcal{A}[t]$ (i.e., no packet it being transmitted on it). For a given {\it idle} edge $(u,v)$ at time $t$, a packet of round $r$ (idle or active) is said to be useful if:
\begin{enumerate}
\item (Aggregation Condition) Both $u$ and $v$ are in $FP_r[t]$.  
\item (Non-isolation Condition) For all $w\in FP_r[t]\cap N^-(u)$, there is an alternate route for aggregation, i.e., $|FP_r[t]\cap N^+(w)|\geq 2$. 
\end{enumerate}
Figure \ref{fig:useful} illustrates the above conditions for determining whether a packet is useful with respect to a link. Note that the definition of valid footprint-sets is consistent with definition of useful packets: by ensuring transmission of only useful packets, we ensure that the footprint of any round must be a valid footprint set (i.e., always containing a spanning tree rooted at $a$).
\begin{figure}[t]
\centering \includegraphics[scale=0.6]{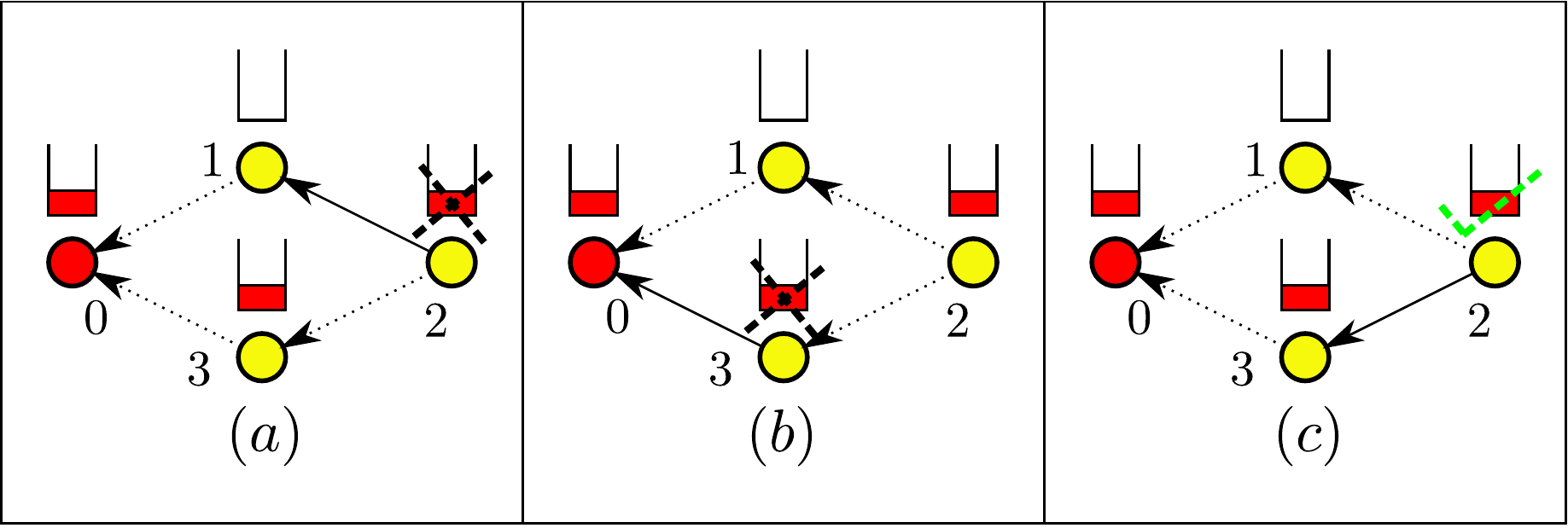}
\caption{Illustration of notion of `useful' packet: For the single round in the above network, packets are not useful for: $(a)$ Link $(2,1)$ because of violating the aggregation condition (node $1$ has no corresponding packet), $(b)$ Link $(3,0)$ because of violating the non-isolation condition (node $2$'s packet gets isolated). $(c)$ The packet is useful for link $(2,3)$.}
\label{fig:useful}
\end{figure}

Next we impose a work conservation requirement on the system in the following manner. Define $X_{+u-v}\triangleq\sum_{S\in\mathcal{S}:v\in S,u\notin S}X_{S+u}$ to be the number of useful idle packets across edge $(u,v)$, and similarly $X^a_{+u-v}$ to be the number of active packets at $u$ which are useful to $v$. Then we impose the following activity condition on the network--$\forall\,(u,v)\in \mathcal{L},$ one of the following is true:
\begin{itemize}
\item[i .]$u\in F \mbox{ for some }(W,F)\in A$
\item[ii.]$X_{+u-v}=0, X^a_{+u-v}=0,$
\end{itemize}
or in words, an edge is active as long as there is at least one useful packet across it. We now describe the routing algorithm, which performs random useful packet forwarding with aggregation while ensuring the activity condition. The routing is performed whenever a link is idle.

%%\vspace{0.1in}
\begin{algorithm}[H]
\label{algo:wlinerpfalgo}
\SetAlgoLined
%\SetLine
\KwIn{An idle link $(u,v)$, i.e., a link with no packet transmitting on it currently.} 
\KwOut{A routing decision of which packet to transmit on $(u,v)$.}
\KwSty{Step 1:} If $\nexists$ useful packets across $(u,v)$, leave link idle.\\
\KwSty{Step 2:} Otherwise, pick a useful packet uniformly at random and start transmitting.\\ 
\caption{Random useful packet forwarding with aggregation for FMux computation in wireline networks.}
\end{algorithm}
%%\vspace{0.1in}

\noindent And finally we have the main theorem for the stability of the algorithm.
\begin{theorem} 
\label{thm:randomopt}
{\corr For a directed acyclic network operating} under algorithm \ref{algo:wlinerpfalgo}, the network is stable if
$\lambda<(\log |\mathcal{R}(f)|)^{-1}\delta^*$, where $\delta^*\triangleq\min_{S\in\mathcal{S}}\sum_{v\in S}\sum_{u\notin S}c_{uv} $.
\end{theorem}

The proof closely follows the proof of Massouli{\'e} et al.\cite{MassTwigg}, with appropriate modifications in order to perform aggregation rather than broadcast. Similar to \cite{MassTwigg}, it proceeds in three stages-
\begin{enumerate}
\item Defining the fluid limit of the Markov chain, and associated convergence results.
\item Defining a Lyapunov function for the fluid system, and showing negative drift.
\item Using the fluid Lyapunov and convergence results to show stability of the original system.
\end{enumerate}
The critical additions that we make are in the appropriate definition of a useful packet, and in identifying the appropriate counter variables that capture FMux aggregation in networks. Further, in Lemma \ref{lem:counting}, we derive a crucial combinatorial relation between these counter variables parallel to the main lemma in \cite{MassTwigg}. The details of the proof are provided in Appendix A.

\section{Scheduling With Aggregation-Tree Routing in Wireless Networks} 
\label{sec:Algo2}

The presence of interference in wireless networks necessitates efficient scheduling of independent sets in addition to routing. Given an SSS rule, we can modify Theorem~\ref{thm:randomopt} to show that random packet aggregation supports a rate upto the min-mincut under the corresponding SSS rule. However dynamic scheduling in order to achieve the optimal SSS rule $\pi^*$ (i.e., with the largest min-mincut) needs an alternate routing technique.

We now describe an alternate approach to throughput-optimal dynamic scheduling and routing for in-network FMux computation over wireless networks. Unlike wired networks, where routing was performed via random packet forwarding, we now focus on schemes based on pre-allocating the route to be followed by the packets of each round, and then scheduling under these routing constraints. Building on the intuition that the ``correct'' routing structures for FMux computation are spanning trees rooted at the aggregator (henceforth refered to as aggregation trees), we split the algorithm into two components:
\begin{itemize}
\item[$\bullet$] A routing component that maps incoming rounds to aggregation trees. Once a round is assigned to a tree, its packets follow the edges of the tree to the aggregator. 
\item[$\bullet$] A scheduling component uses the knowledge of the next hop of each packet to determine an optimal independent set for transmission. 
\end{itemize}
The main result of this section is that there is a dynamic algorithm of this type that is throughput optimal for wireless networks. More specifically, we present a throughput-optimal algorithm based on `greedy' routing (whereby the aggregation tree is chosen in a greedy manner) and `MaxWeight'-type scheduling (whereby links are scheduled according to a maximum weighted independent set problem, with link weights determined by the queues).

Before presenting the general algorithm, we consider some specific example networks to give an intuition as to how the algorithm is constructed; in particular we illustrate the scheduling and routing components separately. Finally, in Section \ref{ssec:FMuxgeneral}, we present the complete algorithm for general graphs, and prove its throughput-optimality.

\subsection{Scheduling With Aggregation-Tree Routing for FMux Computation: Preliminaries and Some Examples} 

In this section we give some examples to build some intuition for the general algorithm we present in Section \ref{ssec:FMuxgeneral}. Suppose the network is a tree rooted at node $a$. For node $i\in\mathcal{N}$, we define $p(i)$ to be the (unique) parent node and $C(i)$ to be the set of immediate children nodes in the aggregator tree. Before specifying the queueing dynamics for this system, we first need a lemma that reduces the space of all possible scheduling policies to a smaller set of policies for which we can write the dynamics in a convenient manner. 

A scheduling policy for tree aggregation is said to be of type \textit{aggregate and transmit} or Type-AT if for every node $i$, and every round $r$, a packet of round $r$ is transmitted from $i$ to $p(i)$ only after receiving the corresponding round $r$ packet from every node $j\in C(i)$. A Type-AT policy thus prevents a round from being transmitted to its parent until it has aggregated all corresponding packets from its children--this is analogous to the non-isolation requirement in Section \ref{sec:Algo1}. Further, this ensures that the flow on each edge of the tree is equal to the input rate of packets on that the tree. Henceforth, we restrict to Type-AT policies, which are sufficient by the following lemma. 

\begin{lemma}
\label{ATisgood}
For an aggregation tree and a scheduling policy that stabilizes the system for given refresh rate $\lambda$, there exists a scheduling policy that stabilizes the system for the same refresh rate, and in addition, is of Type-AT.
\end{lemma}  

\begin{proof}
Given any stabilizing policy, we can use a standard coupling argument to obtain a scheduling policy of Type-AT. Whenever the policy transfers a packet violating non-isolation, the modified algorithm stores the packet at the same node. This continues until the node has received all packets of that round from its children nodes. Now the next time the policy transmits a packet of the same round from that node (which we know happens as the algorithm is stable), the modified algorithm transmits the aggregated packet. However since each round starts off with $|\mathcal{N}|$ packets, this means that the number of packets under the modified algorithm is less than or equal to  $|\mathcal{N}|$ times the number of packets under the non Type-AT algorithm. Since the original policy is stable, hence the modified policy is also stable, and is of Type-AT. 
\end{proof}

We now consider some example networks with $N$ nodes, where the aggregator node $a$ desires a function of the sensed data. Assume that each sensor node records a value from an ordered, finite set $\mathcal{X}$, and the aggregator wants the MAX of these values (an FMux function). The computation at the nodes consists of taking all available packets of a given round, and retaining the one with the largest value. In the following examples, we focus on the routing and scheduling aspects of the problem: first we study how to schedule links to deal with interference under a single aggregation tree; next we allow for collections of aggregation trees with fixed flows and show how to mix flows across these trees; finally we show a simple example of how dynamic routing over many trees can be achieved. In the next section, we combine these to obtain a dynamic scheduling and routing algorithm for general network topologies. 

%%%\vspace{-0.3cm}
\begin{center}
  \begin{tabular}{p{114mm}}
    \hline
    \hline
 	\\
  \end{tabular}
\end{center}
\vspace{-0.3cm}
\textit{Example $3$ (Single Aggregation Tree):} Consider a sensor network where the MAX is computed by combining data on a single aggregation tree. We now modify the queueing model of Section \ref{sec:model} to ensure that a policy is Type-AT. Each node $i$ maintains two queues: $\mathcal{Q}_i^{nu}[t]$ corresponding to `not-useful' packets which are awaiting packets from $C(i)$ with the same round index, and $\mathcal{Q}_i^u[t]$ corresponding to `useful' packets which are ready for transmission to $p(i)$, having received and calculated the MAX of all corresponding packets from nodes in $C(i)$. We also define $\mathcal{Q}_i[t]=\mathcal{Q}_i^u[t]\bigcup\mathcal{Q}_i^{nu}[t]$, and use $Q_i^u[t], Q_i^{nu}[t]$ and $Q_i[t]$ to denote the cardinality of the appropriate queues.

Packets entering the network at node $i$ at time $t$ are stored in $\mathcal{Q}_i^{nu}[t]$ except in leaf nodes where they are stored in $\mathcal{Q}_i^u[t]$. A node only transmits packets which are in $\mathcal{Q}_i^u[t]$ in order to ensure that the policy is of Type-AT. When node $i$ receives packets corresponding to round $r$ from all nodes in $C(i)$, it retains the packet with the maximum value and stores it in $\mathcal{Q}_i^u[t]$. Formally, we can write the queue dynamics as:
\begin{align*}
\mathcal{Q}_i^{nu}[t+1]&=\mathcal{Q}_i^{nu}[t]+\mathcal{A}_i[t]-\mathcal{I}_i[t],\\
\mathcal{Q}_i^u[t+1]&=\mathcal{Q}_i^u[t]+\mathcal{I}_i[t]-\mathcal{D}_{(i,p(i))}[t]. 
\end{align*}
Here $\mathcal{D}_{(i,p(i))}[t]$ represents the packets transmitted from node $i$ to its parent in time slot $t$ and $\mathcal{I}_i[t]$ denotes the internal transfer of packets at node $i$ from unaggregated to aggregated ($r\in \mathcal{I}_i[t]$ if $r\in D_{(j,i)}[t]$ for at least one $j\in C(i)$ and $r\notin \cup_{j\in C(i)}\mathcal{Q}_j[t+1]$). The cardinality of $\mathcal{D}_{(i,p(i))}[t]$ is henceforth denoted as $D_{ip(i)}[t]$ which represents the number of packets transmitted over link $(i,p(i))$ in time slot $t$.  

One observation regarding these dynamics is that unlike data networks, under Type-AT policies, a packet transmission by node $i$ does not change the total size of its parent's queues $Q_{p(i)}[t]$ (this is in general due to the FMux property). Further, each unaggregated round in the network has a useful packet at some node. Thus, we obtain the following scheduling algorithm, which is a modified version of the Backpressure policy \cite{TassEph} to account for these facts:

\begin{algorithm}[H]
\label{algo:compbpalgo}
\SetAlgoLined
%\SetLine
\KwIn{Time slot $t$, queue states $\{\mathcal{Q}_i^u[t],\mathcal{Q}_i^{nu}[t]\}_{i\in\mathcal{N}}$, incoming packets $\mathcal{A}_i[t]$, admissible rate region $\Gamma$}
\KwSty{Step 1:} Place incoming packets to sensor $i$ in $\mathcal{Q}_i^{nu}[t]$ for non-leaf nodes, and $\mathcal{Q}_i^u[t]$ for leaf nodes.\\
%\KwSty{Step 2:} Calculate $Q_i[t]=Q_i^u[t]+Q_i^a[t]$.\\
\KwSty{Step 2:} Compute $\mathbf{c}^*[t]$ as:
\begin{equation*}
\mathbf{c}^*[t]=\arg\max_{\mathbf{c}\in\Gamma}\sum_{i\in\mathcal{N}}Q_i^u[t]c_{ip(i)}[t].
\end{equation*}.\\
\KwSty{Step 3:} Consider node $i$. If $c^*_{ip(i)}[t]>0$ and $Q_i^u[t]>0$, then transmit the first $D_{ip(i)}[t]$ packets, where
\begin{equation*}
D_{ip(i)}[t]=\min(c^*_{ip(i)}[t], Q_i^u[t]).
\end{equation*}
\end{algorithm}

\begin{center}
  \begin{tabular}{p{114mm}}
    \hline
    \hline
 	\\
  \end{tabular}
\end{center}
\vspace{-0.3cm}

The above example indicates how the algorithm chooses independent sets for a single class of packets. Next we consider a network which uses a collection of aggregation trees for routing, therefore requiring the algorithm to make an additional decision of which packet to transmit on a scheduled link.

%%%\vspace{-0.3cm}
\begin{center}
  \begin{tabular}{p{114mm}}
    \hline
    \hline
 	\\
  \end{tabular}
\end{center}
\vspace{-0.3cm}
\textit{Example $4$ (Multiple aggregation trees):} Consider a network modeled by a directed graph where we restrict the routing to a specified collection of aggregation trees. We assume that each tree $\tau$ has a pre-determined arrival rate $\lambda_{\tau}$ of rounds on it. Each new round is associated with a given tree in accordance to the arrival rates, thereby completely specifying the routing. In each time slot, flows from different trees can be scheduled for transmission. We first need some additional notation.

Let $\mathcal{T}$ be the set of spanning trees that are used for routing. Each incoming round is tagged with a specific aggregation tree $\tau\in\mathcal{T}$, which specifies the route to be followed by packets of that round while calculating the MAX at each node. Define $p^\tau(i),C^\tau(i)$ to be the parent and children nodes of node $i$ on tree $\tau$. Also, define $A[t]=\sum_{\tau\in\mathcal{T}}A_i^\tau[t]$ and $\lambda=\sum_{\tau\in\mathcal{T}}\lambda_\tau$ to be the given splitting of input traffic between the aggregation trees. The queueing model is an extension of the previous model. Each node $i$ maintains two queues for each tree $\tau$: $\mathcal{Q}_i^{\tau,nu}[t]$ corresponding to unaggregated packets which are awaiting packets from $C^\tau(i)$ with the same round index (not-useful), and $\mathcal{Q}_i^{\tau,u}[t]$ corresponding to aggregated packets which are ready for transmission to $p^\tau(i)$ (useful). We 
%also define $\mathcal{Q}_i^\tau[t]=\mathcal{Q}_i^{\tau,u}[t]\bigcup\mathcal{Q}_i^{\tau,a}[t]$, and 
use $Q_i^{\tau,u}[t]$ and $Q_i^{\tau,nu}[t]$ to denote the cardinality of the appropriate queues. The queue update equations are similar to before.

The scheduling algorithm for this network is similar to the single aggregation tree, with the added step that the weight of a link is now given by the maximum queue backlog over all queues competing for that link. Formally we have:
%%\vspace{0.1in}
\begin{algorithm}[H]
\SetAlgoLined
%\SetLine
\KwIn{Time slot $t$, queues $\{\mathcal{Q}_i^{\tau,u}[t],\mathcal{Q}_i^{\tau,nu}[t]\}_{i\in\mathcal{N},\tau\in\mathcal{T}}$, incoming packets $\mathcal{A}_i[t]$, admissible rate region $\Gamma$.}
%\KwOut{A scheduling decision $\left\{c_{(i,p^\tau(i))}^\tau[t]\right\}_{i\in\mathcal{N}}$.}
\KwSty{Step 1:} Place incoming packets as before.\\
\KwSty{Step 2:} Calculate $P_{ij}[t]=\max_{\tau\in\mathcal{T}:(i,j)\in\tau}Q_i^{\tau,u}[t]$. Also define $\tau^*(i,j)[t]$ as the tree which maximizes $P_{ij}[t]$.\\
\KwSty{Step 3:} Compute schedule $\mathbf{c}^*[t]$ as:
\begin{equation*}
\mathbf{c}^*[t]=\arg\max_{\mathbf{c}\in\Gamma}\sum_{(i,j)\in\mathcal{L}}P_{ij}[t]c_{ij}[t].
\end{equation*}\\
\KwSty{Step 4:} Consider link $(i,j)$. If $c^*_{ij}[t]>0$, then transmit the first $D_{ij}[t]$ packets of queue $Q_i^{\tau^*(i,j)[t],u}[t]$, where:
\begin{equation*}
D_{ij}[t]=\min(c^*_{ij}[t], Q_i^{\tau^*(i,j)[t],u}[t]).
\end{equation*}
\end{algorithm}

%%%\vspace{-0.3cm}
\begin{center}
  \begin{tabular}{p{114mm}}
    \hline
    \hline
 	\\
  \end{tabular}
\end{center}
\vspace{-0.3cm}

The above two examples indicate how the scheduling algorithm works when the routing is specified. As we mentioned before, the routing component of the algorithm assigns incoming packets to aggregation trees. The challenge is to do so in a dynamic manner, i.e., to route the packets based on network state alone, and not using pre-computed rates for each tree. As we mentioned before, this routing decision is made in a `greedy' manner. In the next example, we consider a simple network to illustrate this.

%%%\vspace{-0.3cm}
\begin{center}
  \begin{tabular}{p{114mm}}
    \hline
    \hline
 	\\
  \end{tabular}
\end{center}
\vspace{-0.3cm}
\begin{figure}[b]
\centering \includegraphics[scale=0.6]{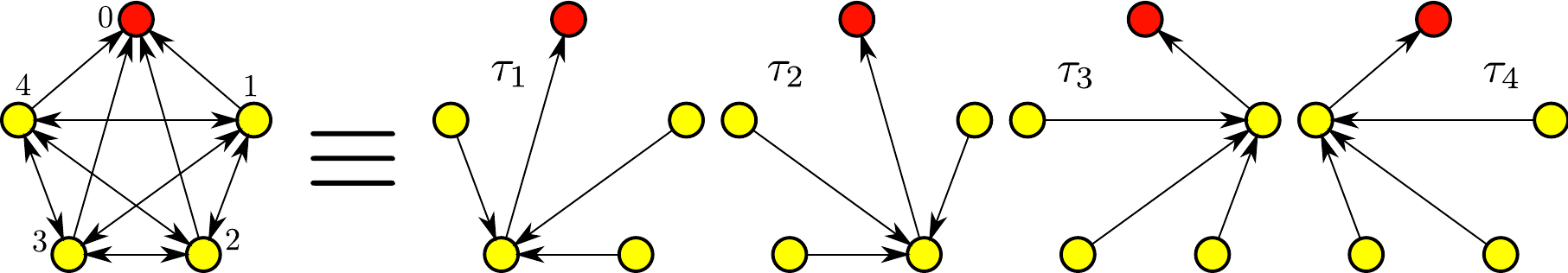}
\label{fig:cliquetrees}
\caption{Decomposition of a complete graph on $5$ nodes into $4$ edge-disjoint aggregation trees.}
\end{figure}
\textit{Example $5$ (Complete Graph):} Consider a network which in the form of a complete graph of $N$ nodes labelled $\{0,1,\ldots,N-1\}$ with node $0$ denoting the aggregator node (which again wants to calculate the MAX value of the data at all the other nodes), and with each link having unit capacity. As we claimed earlier, the min-min-cut of this network can be achieved by packing $N-1$ aggregation trees. In particular, consider the set of depth $2$ trees $\{\tau_i\}_{i=1}^{N-1}$, where tree $\tau_i$ consists of nodes $\{1,3,\ldots,N-1\}\setminus\{i\}$ at the bottom level connected to node $i$ which is connected to the aggregator, i.e., node $0$ (for example, consider the decomposition of a $5$ node complete graph in figure $2$). These trees are clearly edge-disjoint and hence they can each support a load of $1$ to achieve a tree packing of $N-1$ (as for each edge of the graph, there is a single such tree which traverses it. Since all the edges have equal capacity, therefore putting unit capacity on each tree gives us a feasible packing). Hence they are optimal.

There are two ways to route packets on these trees. Since we know that the optimal load on each tree is $1$ unit, we can associate each incoming round of packets to tree $\tau_i$ with probability $\frac{1}{N-1}$. Alternately, when a new round of packets arrives, we can load it on the tree $\tau_i$ that  has the least total number of packets on it. Intuitively this scheme also asymptotically achieves the appropriate load balancing. In the next section, we formalize this notion of `greedy' tree-loading for general graphs, and further show that it indeed does achieve the optimal tree-packing. A more subtle point is that we may not a priori know the correct trees to route on (unlike in this example), and a surprising result is that it is sufficient to perform greedy tree-loading over \textit{all} aggregation trees and still remain throughput-optimal. 
%%%\vspace{-0.3cm}
\begin{center}
  \begin{tabular}{p{114mm}}
    \hline
    \hline
 	\\
  \end{tabular}
\end{center}
\vspace{-0.3cm}

\subsection{Scheduling With Aggregation-Tree Routing for FMux Computation: The General Algorithm} 
\label{ssec:FMuxgeneral}

Finally we present the complete dynamic algorithm for FMux computation. The algorithm separates the routing and scheduling components as follows: when a round of packets arrives in the network, we first `load' all packets of the round on an aggregation tree (thereby fixing the routing); next, in each time slot, scheduling is done according to a modified MaxWeight policy. 

The routing is performed using a {\it greedy tree-loading policy}, wherein all incoming rounds in a time slot are loaded on the tree with smallest sum useful-queue, i.e., least number of useful packets. Formally, we have:

%%\vspace{0.1in}
\begin{algorithm}[H]
\label{algo:treechoosealgo}
\SetAlgoLined
%\SetLine
\KwIn{Time slot $t$, queues $\{\mathcal{Q}_i^{\tau,nu}[t],\mathcal{Q}_i^{\tau,u}[t]\}_{i\in\mathcal{N},\tau\in\mathcal{T}}$, incoming rounds.}
\KwOut{A routing decision associating each incoming round with a tree $\tau\in\mathcal{T}$.}
\KwSty{Step 1:} Calculate $W_{\tau}=\sum_{i\in\mathcal{N}}\left(Q_i^{\tau,u}[t]\right)$ for all $\tau\in\mathcal{T}$.\\
\KwSty{Step 2:} Find the minimum loaded tree $\tau^*[t]$ as:
\begin{equation*}
\tau^*[t]=\arg\min_{\tau\in\mathcal{T}} W_{\tau}[t].
\end{equation*}
\KwSty{Step 3:} Assign all incoming rounds to aggregation tree $\tau^*[t]$. 
\caption{Greedy tree-loading algorithm for FMux computation.}
\end{algorithm}
%%\vspace{0.1in}

The scheduling algorithm is similar to the MaxWeight policy\cite{akrsvw}, in that it picks a maximum independent set with weights given by the product of the rate and the maximum queue across an edge. Formally we have the following algorithm:

%%\vspace{0.05in}
\begin{algorithm}[H]
\label{algo:compbpalgo1}
\SetAlgoLined
%\SetLine
\KwIn{Time slot $t$, queues $\{\mathcal{Q}_i^{\tau,nu}[t],\mathcal{Q}_i^{\tau,u}[t]\}_{i\in\mathcal{N},\tau\in\mathcal{T}}$, incoming packets $\mathcal{A}_i[t]$, admissible rate region $\Gamma$.}
\KwOut{A scheduling decision $\left\{c_{(i,p^\tau(i))}^\tau[t]\right\}_{i\in\mathcal{N}}$.}
\KwSty{Step 1:} Place packets arriving on tree $\tau$ at node $i$ in $\mathcal{Q}_i^{\tau,nu}[t]$ for non-leaf nodes, and $\mathcal{Q}_i^{\tau,u}[t]$ for leaf nodes.\\
\KwSty{Step 2:} Calculate $P_{ij}[t]=\max_{\tau\in\mathcal{T}:(i,j)\in\tau}Q_i^{\tau,u}[t]$. Also define $\tau^*(i,j)[t]$ as the tree which maximizes $P_{ij}[t]$.\\
\KwSty{Step 3:} Compute schedule $\mathbf{c}^*[t]$ as:
\begin{equation*}
\mathbf{c}^*[t]=\arg\max_{\mathbf{c}\in\Gamma}\sum_{(i,j)\in\mathcal{L}}P_{ij}[t]c_{ij}[t].
\end{equation*}\\
\KwSty{Step 4:} Consider link $(i,j)$. If $c^*_{ij}[t]>0$, then transmit the first $\min(c^*_{ij}[t], Q_i^{\tau^*(i,j)[t],u}[t])$ packets from $Q_i^{\tau^*(i,j)[t],nu}[t]$.
\caption{MaxWeight scheduling algorithm.}
\end{algorithm}
%\vspace{0.01in}

For the sake of completeness, we note that in all the above algorithms, tie-breaking rules as well as the service discipline (i.e., among a set of multiple packets suitable for transmission, which one gets priority) are assumed to be random; this is done for the sake of convenience, and we note that there are many possible tie-breaking rules and service disciplines which would suffice.

Now we can state and prove the throughput optimality of this algorithm. 
\begin{theorem}
\label{thm:treeSSS}
The dynamic queue based policy consisting of greedy tree loading (Algorithm \ref{algo:treechoosealgo}) and MaxWeight scheduling (Algorithm \ref{algo:compbpalgo1}) stabilizes the system for any refresh rate $\lambda$ that is less than the maximum refresh rate $\lambda^*$.
\end{theorem}  

Before proceeding further, we point out a particular novel aspect of the proof of this theorem. Similar to previous papers \cite{akrsvw, TassEph}, we use a quadratic Lyapunov function for showing stability; however our technique for bounding the Lyapunov drift is quite different from those used for point-to-point data. The difficulty arises from the fact that although Edmonds' Theorem guarantees the existence of an optimal tree-packing for the network, the trees in this optimal packing are unknown to the algorithm; consequently it is unclear whether routing over all trees could lead to instability via packet accumulation on trees not involved in the optimal packing. We circumvent this by showing the existence of some intermediate tree packings between the optimal and the desired refresh rates, which allow uniform bounding of the Lyapunov drift. We now present the complete proof.

\begin{proof}
We define a candidate Lyapunov function $V[t]$ as
\begin{align*}
V[t]=\sum_{i\in\mathcal{N}}\sum_{\tau\in\mathcal{T}}Q_i^{\tau,u}[t]^2,
\end{align*}
with corresponding Lyapunov drift given by
\begin{equation*}
\Delta_V[t]=\mathbb{E}[V[t+1]-V[t]|\mathbf{Q}[t]].
\end{equation*}
Similar to before, we have that $V[t]\geq 0$ for all states of the system, and that $\Delta_V[t]<\infty$. We now need to show that given $\delta>0$, there exists $Q_{\max}$ such that if $Q_i^{\tau,u}[t]>Q_{\max}$ for some $i$, then $\Delta_V[t]<-\delta$. Now we have
\begin{equation*}
\Delta_V[t]=\sum_{i\in\mathcal{N}}\sum_{\tau\in\mathcal{T}}\mathbb{E}\left[(\Delta Q_i^{\tau,u}[t])^2+2Q_i^{\tau,u}[t]\Delta Q_i^{\tau,u}[t]|\mathbf{Q}[t]\right],
\end{equation*}
and defining $A_i^{\tau,u}[t]$ to be arrival of useful packets on tree $\tau$ to node $i$, we have 
\begin{equation*}
\Delta Q_i^{\tau}[t]=A_i^{\tau,u}[t]-D_{(i,p^{\tau}(i))}^{\tau}[t],
\end{equation*}
and thus $(\Delta Q_i^{\tau}[t])^2\leq m_A + (Lc_{\max})^2$ (due to external arrivals plus inter-node transmissions). Let $M_2=N|\mathcal{T}|(m_A+(Lc_{\max})^2)$. Then we have
\begin{equation*}
\Delta_V[t]\leq M_2+2\sum_{i\in\mathcal{N}}\sum_{\tau\in\mathcal{T}}Q_i^{\tau,u}[t]\mathbb{E}\left[A_i^{\tau,u}[t]-D_{(i,p^{\tau}(i))}^{\tau}[t]|\mathbf{Q}[t]\right].
\end{equation*}

From the definition of $\lambda^*$, we know that there exists an optimal rate point $\{c_{uv}^*\}_{(u,v)\in\mathcal{L}}\in\mathcal{CH}(\Gamma)$ and the corresponding optimal SSS rule $\pi^*$ that maximizes the min-mincut. Consider now a refresh rate $\lambda$ less than the $\lambda^*$, such that $\lambda^*-\lambda=\epsilon>0$. Note that the algorithm can potentially split the incoming flow $\lambda$ over {\it every spanning tree} of the network, in order to dynamically arrive at the optimal packing. To uniformly bound the Lyapunov drift, we first need to construct two tree packings: an `achievable' packing $\{\lambda_{\tau}'\}$ such that $\sum_{\tau\in\mathcal{T}}\lambda_{\tau}'\geq \lambda$ which serves as a proxy for the flow-splitting, and a `near-optimal' packing $\{\hat{\lambda}_{\tau}\}$ such that $\sum_{\tau\in\mathcal{T}}\hat{\lambda}_{\tau}\geq \lambda^*-\frac{2\epsilon}{3}$ and further which has the property that $\hat{\lambda}_{\tau}-\lambda_{\tau}'\geq\epsilon_4$ {\it uniformly over all spanning trees} (for some $\epsilon_4>0$ which we define below). We do so as follows.

Assume that there exists $c_{\min} > 0$ such that if any edge $(u,v)\in\mathcal{L}$ is scheduled alone (i.e. $I'=(u,v)$), then $c_{uv}(I') \geq c_{\min}$ (this is simply a formal definition of
existence of a link). We can now perturb the optimal SSS rule to get a new
rate point $\{\hat{c}_{uv}\}_{(u,v)\in\mathcal{L}}\in\mathcal{CH}(\Gamma)$ with the following two properties:
\begin{enumerate}
\item Every edge $(u,v)\in\mathcal{L}$ has capacity $\hat{c}_{uv}\geq \epsilon_1>0$.
\item The min-mincut of the network at the rate point $\{\hat{c}_{uv}\}_{(u,v)\in\mathcal{L}}$ is $\geq\lambda^*-\frac{\epsilon}{3}$.
\end{enumerate}
This helps ensure that the `near-optimal' tree packing $\hat{\lambda}_{\tau}$ can have some mass on each edge of the graph.

To construct the perturbed SSS rule $\hat{\pi}$, consider the optimal SSS rule $\{\pi^*(I)\}_{I\in\mathcal{I}}$. We define $\mathcal{I}'=\{I\in\mathcal{I}:\pi^*(I)>0\}$ (i.e., the set of independent sets that have some mass under $\pi^*$) and $\pi_{\min}=\min_{I\in\mathcal{I}'}\{\pi^*(I)\}$ (which is $>0$ as the cardinality of $|\mathcal{I}|$ is finite). Now we reduce each $\pi^*(I), I\in\mathcal{I}'$ by $\epsilon_2=\min\{\frac{\pi_{\min}}{2},\frac{\epsilon}{3|\mathcal{L}|c_{\max}} \}$. This reduces the min-mincut by at most $\frac{\epsilon}{3}$. To see this, note that the capacity of any edge $(u,v)$ reduces from $c_{uv}^*$ to $\hat{c}_{uv}$ where:
\begin{align*}
\hat{c}_{uv}&\geq c_{uv}^*(1-\epsilon_2),\\
&=c_{uv}^*-\min\left\{\frac{c_{uv}^*\pi_{\min}}{2},\frac{c_{uv}^*\epsilon}{3|\mathcal{L}|c_{\max}}\right\},\\
&\geq c_{uv}^*-\frac{c_{uv}^*\epsilon}{3|\mathcal{L}|c_{\max}},\\
&\geq c_{uv}^*-\frac{\epsilon}{3|\mathcal{L}|}.
\end{align*}
Further, the maximum number of edges across the min-mincut is bounded by $L$. Thus the min-mincut of the network at the rate point $\{\hat{c}_{uv}\}_{(u,v)\in\mathcal{L}}$ is $\geq\lambda^*-\frac{\epsilon}{3}$.

Next, suppose $\mathcal{L}'$ is the set of edges with zero flow under $\pi^*$. We now complete the definition of the perturbed SSS rule $\hat{\pi}$ (using the fact that singleton edges are valid independent sets) as follows:
\begin{equation*}
\hat{\pi}(I)=\left\{
\begin{array}{l l}
\pi^*(I)-\epsilon_2 & \quad :I\in\mathcal{I}',\\
\frac{|\mathcal{I}'|\epsilon_2}{|\mathcal{L}'|} & \quad :I=\{(u,v)\}\,\forall\,(u,v)\in\mathcal{L}',\\
0 & \quad :\mbox{otherwise}.
\end{array}
\right.
\end{equation*}
To see that this is a valid SSS rule, note that $1-\sum_{I\in\mathcal{I}'}\hat{\pi}(I)=|\mathcal{I}'|\epsilon_2$, which is the weight we have distributed equally over all links in $\mathcal{L}'$. The rate point under this SSS rule is henceforth denoted as $\{\hat{c}_{uv}\}$. Then for edges in $\mathcal{L}'$ we have $\hat{c}_{uv}\geq \frac{|\mathcal{I}'|\epsilon_2 c_{\min}}{|\mathcal{L}'|}$. Now since there are only $L$ edges, each with positive capacity $\hat{c}_{uv}$, therefore there exists some $\epsilon_1>0$ such that every edge $(u,v)\in\mathcal{L}$ has capacity $\hat{c}_{uv}\geq \epsilon_1$ under SSS rule $\hat{\pi}$. Finally, applying Edmonds' Theorem (Theorem~\ref{thm:edmonds}) on the network under $\hat{\pi}$, we get a packing $\{\lambda^*(\hat{\pi})_{\tau}\}_{\tau\in\mathcal{T}}$ such that we have
\begin{equation*}
\sum_{\tau\in\mathcal{T}}\lambda_{\tau}^*(\hat{\pi})\geq\lambda^*-\frac{\epsilon}{3}.
\end{equation*}
Before proceeding further, we need the following definitions:
\begin{itemize}
\item[$\bullet$] $\mathcal{L}^*\triangleq\{(u,v)\in\mathcal{L}:\hat{c}_{uv}-\sum_{\tau:(u,v)\in\tau}\lambda_{\tau}^*(\hat{\pi})=0\}$.
\item[$\bullet$] $\mathcal{T}^*\triangleq\{\tau\in\mathcal{T}:\lambda_{\tau}^*(\hat{\pi})>0\}$.
\item[$\bullet$] $\epsilon_3\triangleq\min_{(u,v)\in(\mathcal{L}^*)^c}\{\hat{c}_{uv}-\sum_{\tau:(u,v)\in\tau}\lambda_{\tau}^*(\hat{\pi})\}$.
\item[$\bullet$] $\lambda_{\min}\triangleq\min_{\tau\in\mathcal{T}^*}\{\lambda_{\tau}^*(\hat{\pi})\}$.
\end{itemize}
Note that $\epsilon_3>0$ as $\hat{c}_{uv}>\epsilon_1$ and the packing is not tight on the finite set $(\mathcal{L}^*)^c$. Similarly, $\lambda_{\min}>0$.

Finally we can construct the tree packings (on the network under SSS rule $\hat{\pi}$) that we need to bound the Lyapunov drift:
\begin{enumerate}
\item The `achievable' tree packing, $\{\lambda_{\tau}'\}_{\tau\in\mathcal{T}}$ is defined as:
\begin{equation*}
\lambda_{\tau}'=\left\{ 
  \begin{array}{l l}
  \max\left\{\lambda_{\tau}^*(\hat{\pi})-\frac{2\epsilon}{3|\mathcal{T}^*|},0\right\} & \quad :\tau\in\mathcal{T}^*, \\
  0 & \quad :\tau\notin\mathcal{T}^*.
  \end{array} \right. 
\end{equation*}
Then clearly $\lambda_{\tau}'$ is a packing (as we are only removing mass from a valid packing) and further: 
\begin{equation*}
\sum_{\tau\in\mathcal{T}}\lambda_{\tau}'\geq\sum_{\tau\in\mathcal{T}}\lambda_{\tau}^*(\hat{\pi})-\frac{2\epsilon}{3}\geq\lambda^*-\epsilon=\lambda .
\end{equation*}
\item The `near-optimal' tree packing, $\{\hat{\lambda}_{\tau}\}_{\tau\in\mathcal{T}}$ is defined as:
\begin{equation*}
\hat{\lambda}_{\tau}=\left\{ 
  \begin{array}{l l}
    \lambda_{\tau}^*(\hat{\pi})-\frac{\epsilon}{3|\mathcal{T}^*|} & \quad :\tau\in\mathcal{T}^*, \lambda_{\tau}^*(\hat{\pi})>\frac{2\epsilon}{3|\mathcal{T}^*|},\\
    \lambda_{\tau}^*(\hat{\pi})-\frac{\lambda_{\min}}{2} & \quad :\tau\in\mathcal{T}^*, \lambda_{\tau}^*(\hat{\pi})\leq\frac{2\epsilon}{3|\mathcal{T}^*|},\\
    \min\left\{\frac{\lambda_{\min}}{2|(\mathcal{T}^*)^c|},\frac{\epsilon}{3|\mathcal{T}^*||(\mathcal{T}^*)^c|},\frac{\epsilon_3}{|(\mathcal{T}^*)^c|}\right\} & \quad :\tau\notin\mathcal{T}^*.
  \end{array} \right. 
\end{equation*}
First we need to show that this is a valid tree packing. To see this, note that the maximum load added on any edge is bounded by $\min\left\{\frac{\lambda_{\min}}{2},\frac{\epsilon}{3|\mathcal{T}^*|},\epsilon_3\right\}$ (since in the worst case, all the trees in $(\mathcal{T}^*)^c$ can contain some edge). For any edge in $(\mathcal{L}^*)^c$, this is less than the slack ($\geq\epsilon_3$ by definition) that was already present. For an edge in $\mathcal{L}^*$, we know at least one tree in $\mathcal{T}^*$ contained it (as every edge in the graph has positive capacity under the SSS rule $\hat{\pi}$), and hence we subtract a load of at least $\min\left\{\frac{\lambda_{\min}}{2|(\mathcal{T}^*)^c|},\frac{\epsilon}{3|\mathcal{T}^*||(\mathcal{T}^*)^c|}\right\}$, which is again greater than the amount of load we add. Thus $\{\hat{\lambda}_{\tau}\}_{\tau\in\mathcal{T}}$ is a valid packing.

Further we have that $\sum_{\tau\in\mathcal{T}}\hat{\lambda}_{\tau}\geq \sum_{\tau\in\mathcal{T}^*}\hat{\lambda}_{\tau}\geq \lambda^*-\frac{2\epsilon}{3}$.\\
\end{enumerate}
In addition, defining
\begin{equation*}
\epsilon_4=\min\left\{\frac{\epsilon}{3|\mathcal{T}^*|},\frac{\lambda_{\min}}{2},\frac{\lambda_{\min}}{2|(\mathcal{T}^*)^c|},\frac{\epsilon}{3|\mathcal{T}^*||(\mathcal{T}^*)^c|},\frac{\epsilon_3}{|(\mathcal{T}^*)^c|}\right\},
\end{equation*}
we get that $\hat{\lambda}_{\tau}-\lambda_{\tau}'\geq\epsilon_4\,\forall\,\tau\in\mathcal{T}$.

Thus we have  constructed the two tree packings we need. We now return to bounding the Lyapunov drift. From above, we have
\begin{equation*}
\Delta_V[t]\leq M_2+2\sum_{i\in\mathcal{N}}\sum_{\tau\in\mathcal{T}}Q_i^{\tau,u}[t]\mathbb{E}\left[A_i^{\tau,u}[t]-D_{(i,p^{\tau}(i))}^{\tau}[t]|\mathbf{Q}[t]\right].
\end{equation*}
Now, let $c_{(i,p^{\tau}(i))}^{\tau}[t]$ be the rate for packets on aggregation tree $\tau$ on link $(i,p^{\tau}(i))$ allocated by the policy in time slot $t$ (thus $\sum_{\tau\in\mathcal{T}}c_{(i,p^{\tau}(i))}^{\tau}[t]=c_{(i,p^{\tau}(i))}[t]$). Then we have
\begin{align*}
\sum_{i\in\mathcal{N}}\sum_{\tau\in\mathcal{T}}Q_i^{\tau,u}[t]\mathbb{E}&\left[D_{(i,p^{\tau}(i))}^{\tau}[t]|\mathbf{Q}[t]\right]\\
&=\sum_{i\in\mathcal{N}}\sum_{\tau\in\mathcal{T}}Q_i^{\tau,u}[t]\mathbb{E}\left[c_{ip^{\tau}(i)}^{\tau}[t]-\max\{c_{ip^{\tau}(i)}^{\tau}[t]-Q^{\tau,u}_i[t],0\}|\mathbf{Q}[t]\right]\\
&\geq\sum_{i\in\mathcal{N}}\sum_{\tau\in\mathcal{T}}Q_i^{\tau,u}[t]\mathbb{E}\left[c_{ip^{\tau}(i)}^{\tau}[t]-\max\{c_{\max}-Q^{\tau,u}_i[t],0\}|\mathbf{Q}[t]\right]\\
&\geq\sum_{i\in\mathcal{N}}\sum_{\tau\in\mathcal{T}}Q_i^{\tau,u}[t]\mathbb{E}\left[c_{(i,p^{\tau}(i))}^{\tau}[t]\right]-NL|\mathcal{T}|c_{\max}^2.
\end{align*}
Further, from the definition of the policy, we know that
\begin{align*}
\sum_{i\in\mathcal{N}}\sum_{\tau\in\mathcal{T}}Q_i^{\tau,u}[t]\mathbb{E}&\left[c_{(i,p^{\tau}(i))}^{\tau}[t]|\mathbf{Q}[t]\right]\\
&=\mathbb{E}\left[\max_{\mathbf{c}\in\Gamma}\sum_{(i,j)\in\mathcal{L}}\max_{\tau\in\mathcal{T}:(i,j)\in\tau}\left\{Q_i^{\tau,u}[t]\right\}c_{ij}[t]|\mathbf{Q}[t]\right]\\
&\geq\max_{\mathbf{c}\in\Gamma}\sum_{(i,j)\in\mathcal{L}}\max_{\tau\in\mathcal{T}:(i,j)\in\tau}\left\{Q_i^{\tau,u}[t]\right\}\mathbb{E}\left[c_{ij}[t]|\mathbf{Q}[t]\right]\\
&\geq\max_{\mathbf{c}\in\mathcal{CH}(\Gamma)}\sum_{(i,j)\in\mathcal{L}}\max_{\tau\in\mathcal{T}:(i,j)\in\tau}\left\{Q_i^{\tau,u}[t]\right\}\mathbb{E}\left[c_{ij}[t]|\mathbf{Q}[t]\right]\\
&\geq\max_{\mathbf{c}\in\mathcal{CH}(\Gamma)}\sum_{(i,j)\in\mathcal{L}}\sum_{\tau\in\mathcal{T}:(i,j)\in\tau}Q_i^{\tau,u}[t]\mathbb{E}\left[c_{ij}^{\tau}[t]|\mathbf{Q}[t]\right]\\
&\mbox{(where $c_{ij}^{\tau}[t]$ is any tree-packing of a given $\mathbf{c}\in\mathcal{CH}(\Gamma)$)}\\
&\geq\sum_{(i,j)\in\mathcal{L}}\sum_{\tau\in\mathcal{T}:(i,j)\in\tau}Q_i^{\tau,u}[t]\hat{c}_{ij}^{\tau},
\end{align*}
where for any edge $(i,j)$, $\hat{c}_{ij}^{\tau}$ represents any valid split of $\hat{c}_{ij}$ between trees lying on that edge, i.e., $\hat{\mathbf{c}}\in\mathcal{CH}(\Gamma)$. In particular, therefore, we can use the tree packing $\{\hat{\lambda}_{\tau}\}_{\tau\in\mathcal{T}}$ to get
\begin{align*}
\sum_{i\in\mathcal{N}}&\sum_{\tau\in\mathcal{T}}Q_i^{\tau,u}[t]\mathbb{E}\left[c_{(i,p^{\tau}(i))}^{\tau}[t]|\mathbf{Q}[t]\right]
\geq\sum_{i\in\mathcal{N}}\sum_{\tau\in\mathcal{T}}Q_i^{\tau,u}[t]\hat{\lambda}^{\tau}.
\end{align*}
Combining inequalities, and defining
$M_3\triangleq N|\mathcal{T}|\left(m_A+(Lc_{\max})^2+Lc_{\max}^2\right),$
we get
\begin{align*}
\Delta_V[t]\leq M_3+2\sum_{i\in\mathcal{N}}\sum_{\tau\in\mathcal{T}}Q_i^{\tau,u}[t]\left(\mathbb{E}\left[A_i^{\tau}[t]|\mathbf{Q}[t]\right]-\hat{\lambda}^{\tau}\right).
\end{align*}

Finally, define $A_i^{\tau}[t]$ to be the rate of rounds arriving on tree $\tau$. Then from the greedy round-tree assignment algorithm, and using the fact that each round results in exactly one useful packet at each node, we get
\begin{flalign*}
\sum_{i\in\mathcal{N}}\sum_{\tau\in\mathcal{T}}Q_i^{\tau,u}[t]&\mathbb{E}\left[A_i^{\tau,u}[t]|\mathbf{Q}[t]\right]\\
&=\sum_{i\in\mathcal{N}}\sum_{\tau\in\mathcal{T}}Q_i^{\tau,u}[t]\mathbb{E}\left[A_i^{\tau}[t]|\mathbf{Q}[t]\right]\\
&=\mathbb{E}\left[\min_{\{A^{\tau}[t]\}:\sum_{\tau\in\mathcal{T}}A^{\tau}[t]=A[t]}\left\{\sum_{\tau\in\mathcal{T}}\left(\sum_{i\in\mathcal{N}}Q_i^{\tau,u}[t]\right)A_i^{\tau}[t]\right\}\Bigg|\mathbf{Q}[t]\right]\\
&\leq\min_{\{A^{\tau}[t]\}}\sum_{i\in\mathcal{N}}\sum_{\tau\in\mathcal{T}}Q_i^{\tau,u}[t]\mathbb{E}\left[A_i^{\tau}[t]|\mathbf{Q}[t]\right]\\
&\leq\sum_{i\in\mathcal{N}}\sum_{\tau\in\mathcal{T}}Q_i^{\tau,u}[t]\lambda_{\tau}'.
\end{flalign*}
Thus we get
\begin{align*}
\Delta_V[t]&\leq M_3-2\sum_{i\in\mathcal{N}}\sum_{\tau\in\mathcal{T}}Q_i^{\tau,u}[t](\hat{\lambda}_{\tau}-\lambda_{\tau}')\\
&\leq M_3-2\epsilon_4\sum_{i\in\mathcal{N}}\sum_{\tau\in\mathcal{T}}Q_i^{\tau,u}[t].
\end{align*}
In order to have $\Delta_V[t]<-\delta$ if $Q_i^{\tau}[t]\geq Q_i^{\tau,u}[t]>Q_{\max}$ for some $i,\tau$, we can choose $Q_{\max}>\frac{M_3+\delta}{2\epsilon_4}$. Thus $V[t]$ is a valid Lyapunov function and by Foster's Theorem, our policy is stabilizing for any $\lambda<\lambda^*$.
\end{proof}

\section{Discussion and Conclusions}
We have presented a queueing-based framework for in-network function computation. We have used this framework to gain insights into designing dynamic and distributed algorithms for in-network function computation and to quantify the performance gains over data-download. We have focused on a class of functions, the FMux functions, which exhibit maximum compression on aggregation, and for which we have used the parity and MAX functions as representative examples. For such functions we have developed scheduling and routing algorithms under different settings. For wireline networks, we have extended the random routing scheme of Massouli{\'e} et al. \cite{MassTwigg} for aggregation. For wireless networks, we have provided a fixed-routing via dynamic flow splitting along with MaxWeight-like scheduling, which is shown to be throughput-optimal. 

The wireless algorithm, as presented, requires routing on all aggregation trees in order to achieve throughput optimality; this may not be practical in many networks due to the potentially exponential number of trees. However, as we showed in the example with the complete graph, one can obtain optimal tree packings with a much smaller number of trees (of the order of $L$) and one direction of future work is to show how such trees can be selected using simple rules in different networks.

%Another interesting topic is exploring the connections between our algorithms and those for wireless multicast \cite{TasSar}. 
Generalizing these algorithms to deal with a broader class of functions, as well as studying the performance of the algorithms with respect to other metrics (delay, energy consumption, among others) are other topics for future work. 

\section*{Acknowledgments}

This work was supported in part by AFOSR under grant FA9550-09-1-0317, NSF Grants CNS-0519535/0519401 (collaborative grant) and CNS-0964391.

% \bibliography{sensorbib}{}
% \bibliographystyle{ieeetr}

\section*{Appendix A\\Scheduling with Random Packet Forwarding: Detailed Proofs}
\label{app:Algo2}

We now present the complete proof for the throughput-optimality of Algorithm \ref{algo:wlinerpfalgo} {\corr in directed acyclic graphs}. Since the proof closely follows the proof of Massouli{\'e} et al.\cite{MassTwigg}, we do not go into complete details, but try mainly to highlight the modifications we make in order to perform aggregation rather than broadcast.

{\corr First we need a lemma that ensures that under the useful packet transmission rule, each round of packets follows a spanning tree. Recall that the footprint of a round of packets is defined as the set of nodes in which the packets of that round is present. Further, recall that a set $S$ is said to be a valid footprint set if each node in $S$ has a path to $a$ in the subgraph induced by $S$; the collection of such sets is denoted by $\mathcal{S}$. We assume throughout that $\mathcal{N}\in\mathcal{S}$, for otherwise the min-mincut is $0$. Note that since we operate in continuous time, only one packet transmission ocurs at a given time with probability $1$; further, we require that the local state information is available at the time of making routing decision. Now we have the following lemma:

\begin{lemma}
For a round of packets with footprint $S\in\mathcal{S}$, the transmission of a useful packet results in a new footprint $S'$ which is also a valid footprint set.
\end{lemma}
\begin{proof}
Since the underlying graph is directed acyclic, we re-label the nodes as $\{0,1,2,\ldots,N-1\}$ according to their topological ordering, where node $0$ is the aggregator $a$, and all edges are from a higher numbered node to a lower numbered node. Further, given a round of packets on a valid footprint set $S$, we have that each node $k\in S$ has at least one route to $a$ using only nodes in $S$; for short, we refer to such a route as a {\it path from $k$ to $a$ in $S$}.

Since we are operating in continuous time, with probability $1$ only one packet transmission occurs at a given time. Now suppose a useful packet is transmitted on edge $(j,i)$, where $i<j$, resulting in a new footprint set $S'=S\setminus\{j\}$. For $S'$ to be a valid footprint, we need that even after the transmission, each node $k\in S'$ has a path to $a$ in $S'$. To do this, we need to consider a partition of the nodes in $S'$ into $3$ classes:
\begin{itemize}
\item  Node $k\in S'$ such that $k<j$ in the topological order: due to the topological ordering property, a path from $k$ to $a$ in $S$ is clearly unaffected by the packet transmission from $j$ to $i$.
\item Node $k\in S', k>j$ such that there exists a path from $k$ to $a$ in $S$ which does not include $j$: such a path is also unaffected by the packet transmission from $j$ to $i$ and hence is still present in $S'$.
\item Node $k\in S', k>j$ such that {\it all paths} from $k$ to $a$ in $S$ pass through $j$: we show by contradiction that this case is impossible under the rules of useful packet transmission. For any path from $k$ to $a$ in $S$, let $k'\leq k$ be the node immediately before $j$ (i.e., the path is $k\rightarrow\ldots\rightarrow k'\rightarrow j\rightarrow\ldots\rightarrow a$). Then $k'$ has no path to $a$ in $S$ that does not pass through $j$, for otherwise we have a path from $k$ to $k'$, and then to $a$, which does not pass through $j$. This means that $k'$ becomes isolated upon transmission of packet from $j$ to $i$, which violates the non-isolation condition of useful packet forwarding.
\end{itemize}
Thus we have that $S'$ is a valid footprint set.
\end{proof}
}
The main idea behind the proof in \cite{MassTwigg} was to define the `footprint counter' variables to represent the state of the system, and considering an appropriate function of these that allowed translating the local decisions of the nodes in terms of global graph parameters. In order to modify the proof for broadcast, we defined a similar collection of counter variables in Section \ref{sec:Algo1}, and now define their associated dynamics as follows.
\begin{itemize}
\item Arrival of new round: $X_{\mathcal{N}}\rightarrow X_{\mathcal{N}}+1$ (This corresponds to adding a packet to the queue with footprint $\mathcal{N}$, as a packet of the new round is simultaneously generated at all the nodes). 
\item Completion of packet transfer: This is only for active packets, i.e., those currently under transmission. For active packet $r\in\mathcal{A}$ with corresponding $(FP_r,E_r)$ and $(u,v)\in E_r$, we have:
\begin{align*}
&FP_r\rightarrow FP_r\setminus\{u\}, E_r\rightarrow E_r\setminus\{(u,v)\},\\
&E_r=\phi\Rightarrow X_{FP_r}=X_{FP_r}+1.
\end{align*}
(The first equation corresponds to removing the edge over which packet transmission was completed, and also updating the footprint of the packet to include the new node. The second updates the list of idle packets in case there is no other instance of this packet being transmitted.).
\item Initiation of a new transfer at an idle link. The new packet is selected uniformly at random among the set of useful packets at the node. If $(u,v)\notin E_r\,\forall\,r\in\mathcal{A}$, then a new packet transfer is formally described as follows:
	\begin{itemize}
	\item Select a useful packet of an idle round with footprint $S\in\mathcal{S}, v\in S, u\notin S$, with probability
	\begin{equation*}
	p_{S}=\frac{X_{S+u}}{X_{+u-v}+X^a_{+u-v}},
	\end{equation*}
	Select a useful packet of an active round $r\in\mathcal{A}$ with $(FP_r,E_r)\in\mathcal{A}$ with probability
	\begin{equation*}
	p_r=\frac{1}{X_{+u-v}+X^a_{+u-v}}.
	\end{equation*}
	\item If idle packet with footprint $S$ is selected: $X_S\rightarrow X_S-1, \mathcal{A}\rightarrow\mathcal{A}\cup\{r\}$, with $(FP_r,E_r)=(S,(u,v))$. If packet of active round $r$ is selected, then $E_r\rightarrow E_r\cup\{(u,v)\}.$
	\end{itemize}
\end{itemize}

We note here that the node itself does not need to know these global counters to perform packet selection; rather, this emerges from the use of the random useful packet forwarding rule. The idea of relating the local packet selection rule to the global counters is crucial in proving the optimality of the algorithm. The local rules for checking whether a packet is useful or not corresponds to selecting packets whose global footprint obeys certain properties; picking a useful packet uniformly at random therefore corresponds to picking a packet from such a useful global footprint with a probability proportional to the corresponding counter variable. 

Observe that in order to determine the flow into a footprint set $S$, we need to consider the collection of sets which include $S$ and have one extra node. We now define the fluid limits of the system. This is similar in spirit to the fluid limit of the system in \cite{MassTwigg}, so we try to use similar notation. The existence of the limit also follows immediately from their convergence results, so we omit it due to lack of space and refer interested readers to \cite{MassTwigg} for technical details. 

\paragraph{The fluid limits of the system:} 
The fluid trajectories $t\rightarrow x_S(t), S\in\mathcal{S}$ corresponding to the system are defined as follows:
\begin{itemize}
	\item[$\bullet$] $\forall\,(u,v)\in \mathcal{L},\,\forall\,S s.t.v\in S,u\notin S, \exists t\rightarrow\phi_{S+u,(u,v)}(t)\,s.t.$
		\begin{align*}
		x_{\mathcal{N}}(t)=&x_{\mathcal{N}}(0)+\lambda t-\sum_{S\in\mathcal{S}:S+u=\mathcal{N} \mbox{ for some } u}\sum_{\mathcal{N}\in S,(u,v)\in \mathcal{L}}\phi_{\mathcal{N},(u,v)}(t)\\
		x_{S}(t)=&x_{S}(0)+\sum_{v\in S}\sum_{u\notin S, (u,v)\in \mathcal{L}}\phi_{S+u,(u,v)}(t)\\
		&-\sum_{S'\in\mathcal{S}:S'+u=S \mbox{ for some } u}\sum_{v\in S':(u,v)\in \mathcal{L}}\phi_{S,(u,v)}(t).
		\end{align*}

	\item[$\bullet$] Work Conservation: At almost every $t, \phi_{S,(u,v)}(t)$ is differentiable and if $x_{+u-v}(t)>0$ (where $x_{+u-v}(t)$ is the fluid trajectory associated with $X_{+u-v}$), then we have
	\begin{align*} 	
	\frac{d\phi_{S+u,(u,v)}(t)}{dt}=c_{uv}\frac{x_{S+u}(t)}{x_{+u-v}(t)}.
	\end{align*} 	

	\item[$\bullet$] $\phi_{S,(u,v)}(t)$ are non-decreasing, Lipschitz continuous, with Lipschitz constant $c_{uv}$, and 				$\sum_{S\in\mathcal{S}:v\in S, u\notin S}\phi_{S+u,(u,v)}(t)$ is $c_{uv}$-Lipschitz.
\end{itemize}

For any $y\in\mathbb{R}^{|\mathcal{S}|}_+, S(y)\triangleq$ set of all fluid trajectories with initial condition $\in C([0,\infty),\mathbb{R}_+^{|\mathcal{S}|})$, and further, we define $\{X_S^N(t)\}_{S\in\mathcal{S}}$ as the state of the MC with initial conditions $(X^N(0),A^N(0))$, $Y_S^N(t)=\frac{X^N_S(z_Nt)}{z_N}$. Now, as in \cite{MassTwigg}, for a sequence of initial conditions $(X^N(0), A^N(0)), N>0$ s.t. for a sequence of positive numbers $(z_N)_{N>0}, \lim_{N\rightarrow\infty}z_N=\infty$ and the limit
\begin{equation*}
\lim_{N\rightarrow\infty}\frac{X^N(0)}{z_N}\triangleq x(0),
\end{equation*}
exists in $\mathbb{R}^{|\mathcal{S}|}_+$, we have that $\forall\,T>0,\epsilon>0$:
\begin{equation*}
\lim_{N\rightarrow\infty}\mathbb{P}[\inf_{f\in S(x(0))}\sup_{t\in[0,T]}||Y^N(t)-f(t)||\geq\epsilon]=0.
\end{equation*}

\paragraph{The fluid Lyapunov function:} Next we define the candidate Lyapunov function that we use to analyze the stability of the system. In \cite{MassTwigg}, the function was defined in terms of queues (or counters) that counted all the packets whose footprint was contained inside a set $S$. The advantage of these queues for studying broadcast was that their rate of increase was controlled by external arrivals to the system, while they were drained due to transfers across the cut defined by the set $S$. 

For the purpose of studying aggregation, we need to identify an equivalent set of queues to reflect the unique dynamics of the system. In particular, we consider for each set $S$ a queue of all rounds whose footprints \textit{are not entirely contained within $S$}. These queues (counters) exhibit similar properties to the ones considered for broadcast in that every incoming round is counted by all these queues (as every node in the network generates a packet), while the drain of these queues is controlled by flow across the cut defined by the set $S$. Formally, we have the following theorem:   
\begin{theorem}
\label{thm:lyap}
Let $\{x_S\}_{S\in\mathcal{S}}$ denote the fluid trajectories. $\forall\,S\in\mathcal{S}$, define:
\begin{equation*}
x_{\nsubseteq S}=\sum_{S'\in\mathcal{S}, S'\nsubseteq S}x_{S'}.
\end{equation*}
Then (given $\lambda, c_{uv}$) $\exists\beta_1, \beta_2,\ldots, \beta_{K-1}>0, \epsilon>0$ such that the Lyapunov function
\begin{equation}
\label{eq:lyapunovfn}
L(\{x_S\}_{S\in\mathcal{S}})\triangleq\max_{S\in\mathcal{S}}\beta_{|S|}x_{\nsubseteq S}
\end{equation}
verifies
\begin{equation}
\label{eq:lyapunovdrift}
L(x(t))\leq\max(0,L(x(0))-\epsilon t).
\end{equation}
\end{theorem}

As in \cite{MassTwigg}, before proving this theorem we first we need a combinatorial lemma. This lemma and its proof parallels a corresponding lemma in \cite{MassTwigg}, with modifications to deal with aggregation and the $x_{\nsubseteq S}$ counter variables we have defined above. %Again, we state and prove the Lemma here for the sake of completeness, closely following the exposition in \cite{MassTwigg}.

\begin{lemma}
\label{lem:counting}
Let $\alpha >0$ be fixed (but arbitrary). We define:
\begin{equation*}
\beta_{i}=\left(1+\frac{1}{\alpha}\right)^{i-1}, i=1,\ldots,K.
\end{equation*}
Then $\forall\,\{x_S\}_{S\in \mathcal{S}}\in\mathbb{R}_+^{|\mathcal{S}|}$, the following conditions hold:
\begin{enumerate}
\item $\forall\,S\in\mathcal{S}, v\in S, u\notin S$, we have
\begin{equation*}
x_{+u-v}<(1+\alpha)^{-1}x_{\nsubseteq S}\Rightarrow \beta_{|S|+1}x_{\nsubseteq S+u}>\beta_{|S|}x_{\nsubseteq S}.
\end{equation*}
\item $\forall\,S\in\mathcal{S}$ such that $\forall\, v\in S, u\notin S, x_{+u-v}\geq (1+\alpha)^{-1}x_{\nsubseteq S}$, if $\exists v\in S, u\notin S$ and some $S'\nsubseteq S, v\in S',u\notin S'$ such that $x_{S'+u}>\alpha x_{+u-v}$, then
\begin{equation*}
\beta_{|S\cup S'|}x_{\nsubseteq S\cup S'}>\beta_{|S|}x_{\nsubseteq S}.
\end{equation*}
\end{enumerate}
\end{lemma}

Note that Lemma \ref{lem:counting} does not depend on the algorithm, or the fluid model in any way. It is a pure combinatorial property of the way that the quantities are defined. In other words, any function mapping the sets $S\in\mathcal{S}$ to $\mathbb{R}_+$ obeys the lemma \textit{for any $\alpha>0$}. Later we use the ability to control $\alpha$ to obtain uniform bounds on the Lyapunov drift.

\begin{proof}
For the first condition, consider $S\in\mathcal{S}, v\in S, u\notin S$ such that
\begin{equation*}
x_{+u-v}<(1+\alpha)^{-1}x_{\nsubseteq S}.
\end{equation*}
Then we have
\begin{align*}
x_{\nsubseteq S}&=x_{\nsubseteq S+u}+x_{S+u},\\
&\leq x_{\nsubseteq S+u}+x_{+u-v},\\
&< x_{\nsubseteq S+u}+(1+\alpha)^{-1}x_{\nsubseteq S},
\end{align*}
and thus
\begin{equation*}
x_{\nsubseteq S}<\frac{1+\alpha}{\alpha}x_{\nsubseteq S+u}.
\end{equation*}
However, from the definition of the $\beta_i$, we have that $\beta_{i}\frac{1+\alpha}{\alpha}=\beta_{i+1}$ for all $i=1,2,\ldots,N-1$. Hence we have that
\begin{equation*}
\beta_{|S|+1}x_{\nsubseteq S+u}>\beta_{|S|}x_{\nsubseteq S}.
\end{equation*}

For the second condition, consider $S\in\mathcal{S}$ such that$\forall\,v\in S, u\notin S, x_{+u-v}\geq (1+\alpha)^{-1}x_{\nsubseteq S}$. Further, consider set $S'$ such that $S'\nsubseteq S, v\in S', u\notin S'$ and satisfying
\begin{equation*}
x_{S'+u}>\alpha x_{+u-v}.
\end{equation*}
Then we have
\begin{align*}
\beta_{|S\cup S'|}x_{\nsubseteq S\cup S'}
&\geq\beta_{|S\cup S'|}x_{S'+u},\\
&\geq\beta_{|S\cup S'|}\alpha x_{+u-v}),\\
&\geq\beta_{|S\cup S'|}\alpha(1+\alpha)^{-1}x_{\nsubseteq S}.
\end{align*}
Thus for our condition, we need
\begin{equation*}
\beta_{|S\cup S'|}\alpha(1+\alpha)^{-1}\geq\beta_{|S|},
\end{equation*}
and noting the fact that $\beta_i$ are increasing with $i$, it is sufficient to ensure
%\begin{equation*}
%\beta_{i+1}\alpha(1+\alpha)^{-1}\geq\beta_{i},\,\forall\,i=1,2,\ldots ,K-1.
%\end{equation*}
%And this we get 
\begin{equation*}
\frac{\beta_{i+1}}{\beta_i}\geq\frac{1+\alpha}{\alpha},\,\forall\,i=1,2,\ldots ,K-1.
\end{equation*}
This in fact holds with equality because of our choice of $\beta_i$. Thus, given any $\alpha>0$, we can construct $\beta_i$ such that the two conditions hold. 
\end{proof}

Now we use Lemma \ref{lem:counting} to prove Theorem \ref{thm:lyap}. The steps of this proof closely follow the corresponding proof in \cite{MassTwigg}.

\begin{proof}
(Proof of Theorem \ref{thm:lyap}) 
Given $\alpha>0$, we define $\beta_i$ as in Lemma \ref{lem:counting}. Then, or any $y\in\mathbb{R}_+^{|\mathcal{S}|}$, if $S^*$ is a set which belongs to arg-max of $\max_{S\in\mathcal{S}}\beta_{|S|}x_{\nsubseteq S}$, then $x_{S^*_+}>0$ (unless all the fluid sample paths are identically $0$).

Next we use the optimality of $S^*$ to obtain some relations between $x_{S^*_+}$ and the weight across its cut-edges. $\forall\,v\in S^*, u\notin S^*$ such that $(u,v)\in \mathcal{L}$, we have from the contrapositive of the first condition of Lemma \ref{lem:counting} (as $S^*$ is in the arg-max) that
\begin{equation*}
x_{+u-v}\geq(1+\alpha)^{-1}x_{\nsubseteq S^*}.
\end{equation*}
Similarly from condition $2$, $\forall\,v\in S^*, u\notin S^*, S'\nsubseteq S^*$ such that $v\in S', u\notin S'$, we have that
\begin{equation*}
x_{S'+u}\leq\alpha x_{+u-v}.
\end{equation*}
Now we have
\begin{align*}
\frac{d}{dt}x_{\nsubseteq S^*}=&\sum_{\nsubseteq S^*}\frac{d}{dt}x_S\\
=&\lambda-\sum_{v\in S^*, u\notin S^*}\sum_{S\subseteq S^*: v\in S}\frac{d}{dt}\phi_{S+u,(u,v)}\\
=&\lambda-\sum_{v\in S^*, u\notin S^*}\left[c_{uv}-\sum_{S'\nsubseteq S^*, v\in S', u\notin S'}\frac{d}{dt}\phi_{S'+u,(u,v)}\right]\\
=&\lambda-\sum_{v\in S^*, u\notin S^*}c_{uv}\left[1-\sum_{S'\nsubseteq S^*, v\in S', u\notin S'}\frac{x_{S'+u}}{x_{+u-v}}\right],\\
&\mbox{(From defn of fluid trajectories)}\\
\leq &\lambda-\sum_{u\in S^*, v\notin S^*}c_{uv}+\sum_{u\in S^*, v\notin S^*}c_{uv}\sum_{S'\nsubseteq S^*, v\in S', u\notin S'}\alpha\\
&\mbox{(From previous observation)}\\
\leq &\lambda-\sum_{v\in S^*, u\notin S^*}c_{uv}+\max_{(u,v)\in \mathcal{L}}c_{uv}|\mathcal{L}|2^K\alpha.
\end{align*}
If we choose $\alpha$ and $\epsilon$ as follows:
\begin{align*}
\alpha=\frac{1}{2}\frac{\delta^*-\lambda}{|\mathcal{L}|2^K\max_{(u,v)\in \mathcal{L}}c_{uv}},\,\ \epsilon=\frac{1}{2}\left(\delta^*-\lambda\right),
\end{align*}
then we get that, for all $S^*\in\arg\max_{S\in\mathcal{S}}\beta_{|S|}x_{\nsubseteq S}$,
\begin{equation*}
\frac{d}{dt}x_{\nsubseteq S^*}\leq-\epsilon\mathds{1}_{y(t)\neq 0}. 
\end{equation*}
To argue that this implies negative drift of the Lyapunov function, i.e. $L(x(t))=\max_{S\in\mathcal{S}}\beta_{|S|}x_{\nsubseteq S}\leq\max(0,L(x(0))-\epsilon t)$, we observe that by definition $\beta_{|S|}\geq 1\,\forall S\in\mathcal{S}$. Finally, using the Lipschitz continuity of the trajectories, it is sufficient to show this property holds for the sets $S^*\in\arg\max_{S\in\mathcal{S}}\beta_{|S|}x_{\nsubseteq S}$.  
\end{proof}

Finally we can prove Theorem \ref{thm:randomopt} using the stability of the fluid limit process along with standard techniques from literature; for technical details, see \cite{MassTwigg}.

\end{document}